\documentclass{ecai} 

\usepackage{latexsym}
\usepackage{amssymb}
\usepackage{amsmath}
\usepackage{amsthm}
\usepackage{booktabs}
\usepackage{enumitem}
\usepackage{graphicx}
\usepackage{xcolor}
\definecolor{pinegreen}{rgb}{0.0,0.47,0.44}
\usepackage{tabularx}
\usepackage{breqn}
\usepackage{algorithm}
\usepackage{algpseudocode}
\usepackage{appendix}

\newtheorem{theorem}{Theorem}

\newtheorem{corollary}[theorem]{Corollary}

\newtheorem{example}{Example}

\newcommand{\BibTeX}{B\kern-.05em{\sc i\kern-.025em b}\kern-.08em\TeX}
\DeclareMathOperator*{\argmax}{arg\,max}

\begin{document}


\begin{frontmatter}

\title{A Mechanism for Mutual Fairness in Cooperative Games  with Replicable Resources -- Extended Version}

\author[A]{\fnms{Björn}~\snm{Filter}\orcid{0009-0008-8666-6239}\thanks{Corresponding Author. Email: bjoern.filter@uni-hamburg.de\\\textbf{Note:} This paper is the extended version of a paper accepted at ECAI 2025, providing the proof of the main theorem in the appendix.} }
\author[A]{\fnms{Ralf}~\snm{Möller}\orcid{0000-0002-1174-3323}}
\author[A]{\fnms{Özgür Lütfü}~\snm{Özçep}\orcid{0000-0001-7140-2574}} 

\address[A]{Institute for Humanities-Centered AI (CHAI), University of Hamburg, Germany}

\begin{abstract}
The latest developments in AI focus on agentic systems where artificial and human agents cooperate to realize global goals. An example is collaborative learning, which aims to train a global model based on data from individual agents. A major challenge in designing such systems is to guarantee safety and alignment with human values, particularly a fair distribution of rewards upon achieving the global goal. Cooperative game theory offers useful abstractions of cooperating agents via value functions, which assign value to each coalition, and via reward functions. With these, the idea of fair allocation can be formalized by specifying fairness axioms and designing concrete mechanisms. Classical cooperative game theory, exemplified by the Shapley value, does not fully capture scenarios like collaborative learning, as it assumes nonreplicable resources, whereas data and models can be replicated. Infinite replicability requires a generalized notion of fairness, formalized through new axioms and mechanisms. These must address imbalances in reciprocal benefits among participants, which can lead to strategic exploitation and unfair allocations. The main contribution of this paper is a mechanism and a proof that it fulfills the property of mutual fairness, formalized by the Balanced Reciprocity Axiom. It ensures that, for every pair of players, each benefits equally from the participation of the other.
\end{abstract}

\end{frontmatter}


\section{Introduction}

The latest developments in artificial intelligence have focused on agentic systems in which artificial and/or human agents interact to achieve global goals. The capabilities of current artificial agents based on large language models have broadened the spectrum of possible interactions between artificial agents or artificial and human agents and have brought about challenging problems in the wider areas of AI safety \cite{hendrycks25AiSafety} and more concretely AI alignment \cite{russell22artificial} and cooperative AI \cite{conitzer23foundationsOFCooperativeAI,conitzer24socialChoice}. 
Even in seemingly unproblematic multi-agent scenarios, such as that of collaborative learning \cite{karimireddy2022mechanisms,qiao23collaborative}, alignment problems must be tackled. One such problem is the topic of this paper: the fair allocation of rewards among agents that provide data for training a global model or database. 

Many solution concepts for the analysis and synthesis of safe and human-value-aligned agentic systems rely on social mechanism design as developed in (computational) choice theory \cite{brandt16handbook} and cooperative game theory \cite{chalkiadakis2011computational}. In particular, cooperative game theory offers useful abstractions that can help solve the problem of fair reward allocation. The first abstraction is that of a value function $v$ that assigns values to coalitions of agents. The value represents the benefit of the coalition relative to values assigned to individual agents. The second abstraction is the reward function, which is meant to describe the distribution of rewards among members of a coalition. A well-known example of a reward function is the Shapley value \cite{shapley53value}. 
With these abstractions, various problems in cooperative game theory, particularly fair allocation, are tackled using the well-proven axiomatic methodology. Axioms formally describe both the constraints on cooperation scenarios and the desired properties of reward functions. This methodology allows for formal treatment of whether a desired reward mechanism exists and whether it is unique. 
 
In classical cooperative game theory, the axiomatic methodology has been applied successfully to fair allocation, resulting in the Shapley reward mechanism. However, the Shapley Value does not fully capture scenarios like collaborative learning, as it assumes non-replicable resources, whereas data and models can be replicated. 
Recent work \cite{sim2020collaborative, wang2020principled, wang2023data} accounts for scenarios with infinitely replicable resources and develops a generalized notion of fairness, formalized through new axioms and mechanisms. However, all of these ignore the relevant aspect of mutual fairness: imbalances in reciprocal benefits among participants can lead to strategic exploitation and unfair allocations. This paper aims to fill this void. 

The paper investigates the implications of the axiom of balanced reciprocity, a variant of which has been considered (for non-replicable resources) in the literature \cite{zou2020sharing}.  This axiom, motivated by the economic concept of fair exchange, ensures that no participant profits from another more than that participant profits in return. The main result of this paper is a concrete reward mechanism induced by an algorithm, along with a proof that it fulfills mutual fairness (formalized by the Balanced Reciprocity Axiom) and other well-known incentivization and fairness axioms. The mechanism ensures that, for every pair of players, each benefits equally from the participation of the other. Moreover, we can show that the mechanism is uniquely determined by the axioms. 

The result applies to a wide variety of cooperative games because its only presumptions on the valuation function $v$ are that the empty coalition is assigned a value of $0$ and that it is monotonic, i.e., larger coalitions yield values at least equal to their sub-coalitions. In particular, the result does not presume stronger properties such as additivity or super-additivity. This is a property broadly presumed in classical cooperative game theory, stating that a coalition's value is not less than the sum of values of its partition's sub-coalitions. In particular, sub-additive games fit collaborative learning scenarios based on infinitely replicable data where the reward function is argued to be concave \cite{karimireddy2022mechanisms}: at some point, replication does not add further value but rather saturates. 

Due to its generality, proving the mutual fairness of the mechanism is not trivial. We sketch the basic proof ideas in the main text and provide the complete proof in the appendix.


\section{Problem Formulation}\label{prob}
In many cooperative settings, multiple participants each contribute individual data points to build a shared dataset. Participants may then be granted access (potentially with restrictions) to the dataset itself or to a machine learning model trained on it. This collaborative arrangement allows all contributors to benefit from the pooled resource. Each agent (or player) possesses data that they can either use independently or contribute to a coalition. The value that a player's data adds to a coalition—and thus the coalition's overall value increase from that contribution—can depend on various factors, including redundancy, complementarity, diminishing returns, and dependencies on prior knowledge. As a result, the marginal value of a player's data points may vary depending on the existing composition of the dataset.

This setting reflects a wide range of real-world applications, such as scientific research, market forecasting, decentralized sensing, and AI development. Here, the usefulness of new data often depends on the information already available. In these contexts, collaboration can enhance collective benefit, but diminishing returns are common: larger coalitions tend to gain less from the addition of new participants than smaller ones, as much of the valuable data may already be present. Consequently, the value of a coalition is not simply additive but shaped by complex interactions among the individual contributions of its members.

An example of a collaborative learning setting is collaborative training of machine learning models, such as those for handwriting recognition, text classification, or metadata extraction, using digitized historical documents held by multiple institutions. Libraries, archives, and research centers often possess valuable manuscript collections and are willing to share their data in exchange for access to a larger model or database than they could build on their own. Collaborative learning enables these institutions to jointly develop a shared model, leveraging the collective value of their datasets. In such a setting, the contribution of each participant to the quality of the final model depends not only on the amount of data available locally but also on its uniqueness, diversity, and relevance. For example, training data that covers underrepresented scripts, languages, or time periods can improve the model more significantly than redundant or homogeneous data. As a result, the overall utility of a coalition in collaborative learning reflects complex interactions among the contributions of the participants. 

Numerous similar settings exist where collaborative learning can facilitate collaboration among data holders. In healthcare and medical research, for instance, hospitals or clinics may jointly train diagnostic models (e.g., for tumor detection or patient risk stratification). In genetic research, collaborative learning allows hospitals, research labs, or pharmaceutical companies to train models for disease prediction, drug response, or population health studies using genomic and clinical data distributed across jurisdictions. In financial market analysis, institutions such as banks, hedge funds, or regulatory bodies might collaborate to improve risk models or anomaly detection systems. Similarly, in the legal domain, courts or legal tech companies may contribute case metadata or annotated legal texts to train natural language processing systems for tasks such as case classification, precedent retrieval, or legal argument mining.

These examples share a common structure: valuable but decentralized datasets and a mutual interest in a stronger shared model. In all such cases, the value of each participant’s contribution depends not just on the quantity but also on the extent to which their data complements those of others. Collaborative learning in these settings forms a natural cooperative game, where coalitional value arises from both the volume and the informational diversity of the data brought together. Furthermore, all these settings have in common that further information does not decrease the value of a dataset, i.e., the value function is monotonic. 

We will now describe how to model our setting formally as a cooperative game. Let $N \subset \mathbb{N}$ be a finite set of players, and we define $|N| = n$. Let a \emph{game} $v$ be a real-valued function $v:2^N \rightarrow\mathbb{R}$ that satisfies $v(\emptyset) = 0$. We call a set $C \subseteq N$ a \emph{coalition} and $v(C)$ the \emph{value} of that coalition. To simplify the notation, we use $v_C$ for $v(C)$, and, in the case of singular member coalitions, we write $v_i$ for $v(\{i\})$. We presume a \emph{monotonic game}, that is, for all $C \subseteq N$ and all $C' \subseteq C$, $v(C') \leq v(C)$. 
More constrained classes that we mention in the text are super-additive games, i.e.,  $v(C_1) + v(C_2) \leq v(C_1 \cup C_2)$ for  $C_1 \cap C_2 = \emptyset$,  sub-additive games, i.e., $v(C_1) + v(C_2) \geq v(C_1 \cup C_2)$ for  $C_1 \cap C_2 = \emptyset$, and  additive games, i.e., $v(C_1) + v(C_2) = v(C_1 \cup C_2)$ for $C_1 \cap C_2 = \emptyset$. Let $\Gamma^N$ denote the set of all games $v$ with the set of players $N$.

In our setting, a \emph{solution function} or, alternatively,  a \emph{reward function/mechanism}  prescribes an element of $\mathbb{R}^{n \times 2^n}$ for each game $v \in \Gamma^N$. This function maps coalitions to the rewards agents receive. Each column corresponds to a possible coalition $C \subseteq N$ while each row corresponds to a player $i \in N$. For such a solution function $\mathcal{M}(v) :=$
\begin{align}\label{defsolfun}
    \begin{bmatrix}
    \mathcal{M}(v)^{\emptyset}_1 & \mathcal{M}(v)^{\{1\}}_{1} & \mathcal{M}(v)^{\{2\}}_{1} & \dots & \mathcal{M}(v)^C_1 & \dots & \mathcal{M}(v)^N_{1} \\
    \vdots & \vdots & \vdots & \ddots & \vdots & \ddots & \vdots \\
    \mathcal{M}(v)^{\emptyset}_n & \mathcal{M}(v)^{\{1\}}_n & \mathcal{M}(v)^{\{2\}}_n & \dots & \mathcal{M}(v)^C_n & \dots & \mathcal{M}(v)^N_n
    \end{bmatrix}
\end{align}
and any $C \subseteq N$, $i \in N$, let $\mathcal{M}(v)^C_i$ denote the reward obtained by player $i$ if coalition $C$ forms. We will write $\mathcal{M}^C_i$ for $\mathcal{M}(v)^C_i$ when the function $v$ is clear from the context. We will call $\mathcal{M}^C_i$ the reward of the player $i$ if coalition $C$ forms.

In this approach, the solution function assigns a reward to every player, including those who are not part of a coalition. We assume that these players operate independently and obtain the value of their individual datasets. Alternatively, one might treat the solution function as a mechanism that interacts solely with the coalition, disregarding players outside the coalition and assigning them a value of zero. However, we decided to adopt the former approach because it accounts for the outcome of every player, not just those within a coalition. Furthermore, we note that the first $n+1$ columns in $\mathcal{M}(v)$ are redundant. These correspond to cases where either no coalition forms or a coalition consists of a single player, both of which yield the same outcome, where each player is working independently. We decided to stick to our definition because it makes the whole approach more homogeneous and eases the notation of the axioms.


\section{Axioms for Incentivization and Fairness}\label{axioms}
An effective reward allocation mechanism in the described setting has to both incentivize players to participate and reward them fairly. To ensure this, we establish a set of axioms that govern valid reward distributions. We begin by defining incentive constraints that ensure feasibility, efficiency, and individual rationality. These constraints guarantee that the solution function remains viable and encourages participation. Building on this foundation, we introduce fairness constraints to prevent unjust disparities in reward distribution. The following sections formalize these principles and their implications.

\subsection{Incentive Axioms}
To formulate a valid solution function, we state the following incentive constraints, ensuring that a solution function is feasible, efficient, and individually rational. These are common solution concepts from cooperative game theory as described in relevant textbooks \cite[Chapter 12]{shoham2008multiagent}, \cite{chalkiadakis2011computational}.  
For any $v \in \Gamma^N$ fulfilling the conditions stated in Section \ref{prob}, a valid solution function $\mathcal{M}(v)$ must fulfill the following axioms:\\
\\
\begin{tabularx}{\linewidth}{l X}
        R1 & \textbf{Non-negativity:} Each player must get a non-negative reward: $\forall C \subseteq N$, $\forall i \in C: \mathcal{M}^C_i \geq 0$.
\end{tabularx}
\begin{tabularx}{\linewidth}{l X}
        R2 &\textbf{Feasibility:} The reward for each player in any coalition $C \subseteq N$ cannot be larger than the value achieved by that coalition: $\forall C \subseteq N$, $\forall i \in C: \mathcal{M}^C_i \leq v_C$.
\end{tabularx}
\begin{tabularx}{\linewidth}{l X}
        R3 & \textbf{Weak Efficiency:} In each coalition $C \subseteq N$, the reward received by at least one player $i \in C$ must be as large as the total value that coalition $C$ can achieve: $\forall C \subseteq N\; \exists i \in C: \mathcal{M}^C_i = v_C$.
\end{tabularx}
\begin{tabularx}{\linewidth}{l X}
        R4 & \textbf{Individual Rationality:} Each player must receive a reward that is at least as large as the value that player can achieve by themselves: $\forall C \subseteq N$, $\forall i \in N: \mathcal{M}^C_i \geq v_i.$
\end{tabularx}
\begin{tabularx}{\linewidth}{l X}
        R5 & \textbf{Non-participation:} For any $C \subseteq N$, the reward received by players outside the coalition $i \notin C$  equals the value they achieve by themselves: $\forall C \subseteq N$, $\forall i \notin C: \mathcal{M}^C_i = v_i$.
\end{tabularx}\\
\\
R1 and R4 are the same axioms for solution concepts as in cooperative game theory with non-replicable rewards. R2 and R3 have been adapted from Chalkiadakis and colleagues \cite{chalkiadakis2011computational}, since in our setting we can give every member of a coalition $C\subseteq N$ a reward of up to $v_C$. We only require weak efficiency, since strong efficiency would imply having to pay out $v_C$ to every member of the coalition, which would maximize welfare, but impede any fairness considerations. Axiom R5 is necessary in our setting, since a solution function describes rewards for every player regardless of whether they are in the coalition. Thus, it has to be ensured that they are treated as if they work alone.

\subsection{Fairness Axioms}

To ensure a fair allocation of rewards, the solution function must meet the following four fairness axioms:\\

\noindent
\begin{tabularx}{\linewidth}{l X}
    F1 & \textbf{Uselessness:} If player $u$'s data does not provide an increase to the value of any coalition, then player $u$ should always receive a valueless reward. Furthermore, all other players in a coalition with $u$ should get identical rewards as they would if $u$ were not part of the coalition.  For all $u \in N$,
\end{tabularx}
\begin{align}
        \begin{split}
            & \left( \forall C \subseteq N \setminus \{u\}: v_C = v_{C \cup \{u\}} \right)\\
            \Rightarrow & \biggl( \left(\forall C \subseteq N: \mathcal{M}^C_u = 0 \right)\\
            & \land \left( \forall C \subseteq N \setminus \{u\}, \forall i \in C: \mathcal{M}^C_i = \mathcal{M}^{C \cup \{u\}}_i \right) \biggr).
        \end{split}
    \end{align}
\begin{tabularx}{\linewidth}{l X}
    F2 & \textbf{Symmetry:} If players $i$ and $j$ contribute identically to any coalition, they should receive equal rewards: For all $i,j \in N$ s.t. $i\neq j$,
    {\begin{align}
        \begin{split}
            & \left( \forall C \subseteq N \setminus \{i, j\}: v_{C \cup \{i\}} = v_{C \cup \{j\}} \right)\\
            \Rightarrow & \left(\forall C \subseteq N \text{ with } i, j \in C: \mathcal{M}^C_i = \mathcal{M}^C_j\right).
        \end{split}
    \end{align}}
\end{tabularx}
\begin{tabularx}{\linewidth}{l X}
    F3 & \textbf{Strict Desirability:} If the value of at least one coalition improves more by including player $i$ instead of player $j$, but the reverse is not true, player $i$ should receive a more valuable reward than $j$. For all $i,j \in N$ s.t. $i\neq j$, 
    {\begin{align}
        \begin{split}
            & \left( \forall A \subseteq N \setminus \{i, j\}: v_{A \cup \{i\}} \geq v_{A \cup \{j\}} \right) \\
            \Rightarrow & \biggl( \forall C \subseteq N \text{ with } i, j \in C \text{ and } \exists B \subseteq C \setminus \{i, j\}, B \neq \emptyset,\\
            & \text{ s.t. } v_{B \cup \{i\}} > v_{B \cup \{j\}} ): \left( \mathcal{M}^C_i > \mathcal{M}^C_j\right) \biggr).
        \end{split}
    \end{align}}
\end{tabularx}\\

\noindent
\begin{tabularx}{\linewidth}{l X}
    F4 & \textbf{Strict Monotonicity:} Suppose the value of at least one coalition containing player $i$ improves (for example, by including more data of $i$), ceteris paribus, then player $i$ should receive a more valuable reward than before: Let $v$ and $v'$ denote any two value functions over all coalitions $C \subseteq N$ and $\mathcal{M}(v)^C_i$ and $\mathcal{M}(v')^C_i$ be the corresponding values of rewards received by player $i$ in coalition $C$. For all $C \in N$ and all $i \in C$,
    {\begin{align}
        \begin{split}
            &\biggl( \left( \forall A \subseteq C \setminus \{i\}: v'_{A \cup \{i\}} \geq v_{A \cup \{i\}} \right)\\
             \land &\left(v'_C > v_C \right) \land \left( \forall D \subseteq C \setminus \{i\}: v'_D = v_D \right) \biggr)\\
            \Rightarrow  & \left( \mathcal{M}(v')^C_i > \mathcal{M}(v)^C_i \right).
        \end{split}
    \end{align}}
\end{tabularx}

F1 and F2 are axioms of the Shapley value \cite{shapley1953value}. Axiom F3 was first introduced by Maschler and Peleg \cite{maschler1966characterization} and, in our setting, reduces to the fact that players who contribute larger values should receive larger rewards, compared to players with less valuable contributions. F4 is adapted from Young \cite{young1985monotonic}.

Together, these axioms were introduced by Sim and colleagues \cite{sim2020collaborative} as conditions for fairness. 
However, our axioms F3 and F4 are slightly weaker than theirs. In F3, we require $B \neq \emptyset$, that is, if two players $i$ and $j$ bring the same increase to every nonempty coalition, they may get the same reward in $C$, even if $v_i > v_j$. In F4, instead of $v'_C > v_C$, Sim and colleagues only required $\exists B \subseteq C \setminus \{i\}: v'_{B \cup \{i\}} > v_{B \cup \{i\}}$, that is, if there is some set for which $B$ $v'_{B \cup \{i\}} > v_{B \cup \{i\}}$, $i$ should get a higher reward not only in $B$ but also in coalitions containing $B$. However, we only require a larger reward for $i$ directly in that $B$. For additive or super-additive $v$, our axioms are equal to those of Sim and colleagues (because an increased value for some coalition $C$ would increase the value of any coalition containing $C$ as well). However, for sub-additive $v$, the original axioms can be incompatible with our F5 which we will formulate next. Therefore, this weakening is necessary.

We expand these fairness conditions by our so-called Balanced Reciprocity Axiom. It states that no player should profit from another player more than that player profits from the first player. If one player contributes significantly to another player's rewards but does not receive a comparable benefit in return, this creates an imbalance that can lead to strategic exploitation. The Balanced Reciprocity Axiom ensures that rewards are assigned proportionally to mutual benefit, preventing situations where a player can make a free ride on another player's contributions without fair compensation.

From an economic perspective, this axiom aligns with the idea of fair exchange: Just as in trade, when two parties interact, neither should end up significantly better off at the expense of the other. It discourages asymmetrical dependencies and ensures a more stable and cooperative environment.

This axiom also ensures that no player has more power over another: If one player threatens to leave, the other is not disproportionately affected. 
This prevents unfair dependencies where one player can coerce or manipulate another by leveraging their importance to the coalition. In cooperative settings, the power of a player comes from how much their presence affects the outcomes for others. If player $i$ benefits player $j$ significantly more than $j$ benefits $i$, then $i$ can use this imbalance as leverage. In extreme cases, $i$ can threaten to leave the coalition knowing that $j$ has much more to lose, forcing $j$ to make concessions.

We formulate this axiom as follows:\\
\\
\begin{tabularx}{\linewidth}{l X}
    F5 & \textbf{Balanced Reciprocity:} Any agent $i$ should profit from the participation of any other agent $j$ in the same way that $j$ profits from the participation of $i$. For all $C \subseteq N$ and all $i, j \in C$:
    {\begin{align}
        \mathcal{M}^C_i - \mathcal{M}^{C \setminus \{j\}}_i = \mathcal{M}^C_j - \mathcal{M}^{C \setminus \{i\}}_j,
    \end{align}}
\end{tabularx} \\
We consider this axiom as a fairness condition, and hence name it as the fifth fairness axiom F5. For any game $v \in \Gamma^N$ that fulfills the conditions stated in Section \ref{prob}, for a function $\mathcal{M}(v)$ to be a valid solution, the following must hold:\\
\\
\begin{tabularx}{\linewidth}{l X}
    R6 & \textbf{Fairness:} For all $C \subseteq N$ and $i \in N$, the rewards $\mathcal{M}^C_i$ must satisfy F1 to F5.
\end{tabularx} \\

\subsection{Uniqueness of the Solution}

Before we present our solution function for R1 to R6, we first show that if such a solution function exists, it must be unique:

\begin{theorem}\label{unique}
    Suppose that for some (not necessarily monotone) $v$ there exists a solution $\mathcal{M}(v)$ that fulfills axioms R3, R5 and F5. Then this solution is unique.
\end{theorem}

\begin{proof}
    We prove Theorem \ref{unique} by contradiction. Let $v \in \Gamma^N$ be an arbitrary but fixed function. Suppose that there are two different solutions $\mathcal{M}(v)$ and $\mathcal{M'}(v)$ that both satisfy R3, R5 and F5. Further, suppose that for some $i \in N$ and some $C \subseteq N$, $\mathcal{M}(v)^C_i \neq \mathcal{M'}(v)^C_i$. We will see that this is a contradiction.

    Let $A \subseteq N$ be the smallest coalition, for which there is some $a \in N$ so that $\mathcal{M}^A_a \neq \mathcal{M'}^A_a$. W.l.o.g. we will assume $\mathcal{M}^A_a < \mathcal{M'}^A_a$. If several such $A$ exist, pick any. Note that this implies that for all $B \subset A$ and all $i \in N$, $\mathcal{M}^B_i = \mathcal{M'}^B_i$.
    
    First, we note that due to R5, for all $C \subseteq N$ and all $i \notin C$, we have $\mathcal{M}^C_i = \mathcal{M'}^C_i = v_i$. Therefore, we must have $a \in A$. Furthermore, due to R3, for all $C \subseteq N$ with $\lvert C \rvert = 1$ and all $i\in C$, we have
    \begin{align}
        \mathcal{M}^C_i = \mathcal{M'}^C_i = \mathcal{M}^{\{i\}}_i = \mathcal{M'}^{\{i\}}_i = v_i = v_C.
    \end{align}
    Therefore, $|A| \geq 2$ must hold.

    Let $k$ be the player such that $\mathcal{M}_k^A = v_A$. Due to R3, such a player must exist. $k$ must be different from $a$, because if $k = a$, $\mathcal{M}^A_a < \mathcal{M'}^A_a$ would imply $\mathcal{M'}^A_k > \mathcal{M}^A_k = v_A$, which would violate R3. Thus,  $k \neq a$.

    F5 yields:
    \begin{align}
        \mathcal{M}^A_k = \mathcal{M}^A_a - \mathcal{M}^{A \setminus \{k\}}_a + \mathcal{M}^{A \setminus \{a\}}_k
    \end{align}
    and
    \begin{align}
        \mathcal{M'}^A_k = \mathcal{M'}^A_a - \mathcal{M'}^{A \setminus \{k\}}_a + \mathcal{M'}^{A \setminus \{a\}}_k.
    \end{align}
    Remember that for all $B \subset A$ and all $i \in N$, $\mathcal{M}^B_i = \mathcal{M'}^B_i$. Therefore, $\mathcal{M}^{A \setminus \{k\}}_a = \mathcal{M'}^{A \setminus \{k\}}_a$ and $\mathcal{M}^{A \setminus \{a\}}_k = \mathcal{M'}^{A \setminus \{a\}}_k$. But then $\mathcal{M'}^A_a > \mathcal{M}^A_a$ entails 
     $   \mathcal{M'}^A_k > \mathcal{M}^A_k = v_A $. 
    But then $\mathcal{M'}(v)$ does not satisfy R3, which is a contradiction to our assumption.
\end{proof}


\section{A Solution Function for Balanced Reciprocity}
We present a solution function that satisfies the axioms of Section \ref{axioms}, in particular F5, by explicitly balancing reciprocity among  participants. The intuition behind our approach is simple: The rewards are computed recursively, for $C \subseteq N$, the rewards depend on the rewards handed out in all $C' \subset N$ with $|C'| = \lvert C \rvert - 1$. For each $C \in N$, it is first examined which agent $k \in C$ should receive the maximum reward $v_C$. Then, all other agents $i \in C \setminus \{k\}$ are assigned their rewards according to the Balanced Reciprocity Axiom F5. The solution function is given in Algorithm \ref{solfun}.

\begin{algorithm}
\caption{Solution Function Algorithm}\label{solfun}
\begin{algorithmic}[1]
    \For {$i \in N$}
        \State {$\mathcal{M}^{\emptyset}_i = v_i$};
        \For {$C \subseteq N$, $\lvert C \rvert = 1$}
            \State {$\mathcal{M}^C_i = v_i$};
    \EndFor
    \EndFor
    \For {$s \in 2:n$ in ascending order}
        \For {$C \subseteq N$, $\lvert C \rvert = s$}
            \For {$i \notin C$}
                \State {$\mathcal{M}^C_i = v_i$};
            \EndFor
            \State {choose any $j \in C$};
            \State {$m^C_j = v_C$};
            \For {$i \in C, i \neq j$}
                \State {$m^C_i = m^C_j - \mathcal{M}^{C \setminus \{i\}}_j + \mathcal{M}^{C \setminus \{j\}}_i$};
            \EndFor
            \State {$k = \argmax_{i \in C} m_i^C$};
            \State {$\mathcal{M}^C_k = v_C$};
            \For {$i \in C, i \neq j$}
                \State {$\mathcal{M}^C_i = \mathcal{M}^C_k - \mathcal{M}^{C \setminus \{i\}}_k + \mathcal{M}^{C \setminus \{k\}}_i$};
            \EndFor
        \EndFor
    \EndFor

\end{algorithmic}
\end{algorithm}

Indeed, the solution function is the reward mechanism that we were looking for: It fulfills all standard axioms of incentivization and fairness, as well as the Balanced Reciprocity Axiom (Theorem \ref{thm:r1to6}). 

\begin{theorem}\label{thm:r1to6}
    The solution function described by Algorithm \ref{solfun} produces a solution $\mathcal{M}(v)$, which for any $v$ that is monotone with $v_{\emptyset} = 0$, fulfills R1 to R5, as well as F1 to F5.
\end{theorem}

\textbf{Sketch of proof:} For each of the axioms, we will give a sketch of the proof that it is fulfilled by Algorithm \ref{solfun}. Detailed proofs are given in Appendix \ref{proof:axioms}. Due to interdependencies of the proofs, we start by proving F5 and R2.\\

F5 can be shown by induction on the size of $C$. For any $C \in N$ with $\lvert C \rvert \leq 2$ it is easy to see that F5 must hold, since either there are not enough players in $C$ for it to apply, or the two players are directly assigned their reciprocal values in line 20. For $\lvert C \rvert > 2$, we can use the fact, that due to the inductive hypothesis, for $i, j, k \in C$, with $k$ being the player picked by the algorithm in line 20, the following equations must hold:
\begin{align}\label{e1}
        \mathcal{M}^{C \setminus \{i\}}_k - \mathcal{M}^{C \setminus \{i, j\}}_k = \mathcal{M}^{C \setminus \{i\}}_j - \mathcal{M}^{C \setminus \{i, k\}}_j.
\end{align}
\begin{align}\label{e2}
        \mathcal{M}^{C \setminus \{j\}}_i - \mathcal{M}^{C \setminus \{j, k\}}_i = \mathcal{M}^{C \setminus \{j\}}_k - \mathcal{M}^{C \setminus \{i, j\}}_k.
\end{align}
\begin{align}\label{e3}
        \mathcal{M}^{C \setminus \{k\}}_i - \mathcal{M}^{C \setminus \{j, k\}}_i = \mathcal{M}^{C \setminus \{k\}}_j - \mathcal{M}^{C \setminus \{i, k\}}_j
\end{align}

Starting with the fact, that due to line 20, $\mathcal{M}^C_i = \mathcal{M}^C_k - \mathcal{M}^{C \setminus \{i\}}_k + \mathcal{M}^{C \setminus \{k\}}_i$ and then substituting terms according to the equations above leads to
\begin{align}
    \mathcal{M}^C_i - \mathcal{M}^{C \setminus \{j\}}_i  = \mathcal{M}^C_j - \mathcal{M}^{C \setminus \{i\}}_j
\end{align}
Since this holds for all pairs $i, j \in C$ with $i \neq j$, the algorithm fulfills F5 for $C$.\\

R2: For any $C \in N$ with $\lvert C \rvert \leq 1$, it is easy to see that R2 must hold, since either there are not enough players in $C$ or there is one $i \in C$ with $\mathcal{M}^C_i = v_i = v_C$.

For larger $C$, if $i$ is the player $k$ chosen according to line 17, $\mathcal{M}^C_i = \mathcal{M}^C_k = v_C$ ensures R2. If $i \neq k$, since $k$ was chosen over $i$ in line 17, $m^C_k \geq m^C_i$, therefore, we have either directly $\mathcal{M}^{C \setminus \{k\}}_i \leq \mathcal{M}^{C \setminus \{i\}}_k$. $\mathcal{M}_i^C$ (if $i$ is the player $j$ picked in line 12, in relation to whom $m^C_k$ is calculated) or $\mathcal{M}^{C \setminus \{j\}}_k - \mathcal{M}^{C \setminus \{k\}}_j \geq \mathcal{M}^{C \setminus \{j\}}_i - \mathcal{M}^{C \setminus \{i\}}_j$. Here, substituting terms again using the equations \ref{e1} to \ref{e3}, yields $\mathcal{M}^{C \setminus \{k\}}_i \leq \mathcal{M}^{C \setminus \{i\}}_k$.

Together with the fact that $\mathcal{M}^C_k = v_C$ (line 18) and the assignment from line 20: $\mathcal{M}^C_i = \mathcal{M}^C_k - \mathcal{M}^{C \setminus \{i\}}_k + \mathcal{M}^{C \setminus \{k\}}_i$, this gives $\mathcal{M}^C_i \leq \mathcal{M}^C_k = v_C$ and therefore the algorithm fulfills R1 for $C$.\\

R1 can again be shown by induction over the size of $C$. For any $C \in N$ with $\lvert C \rvert \leq 2$ it is easy to see that R1 must hold, since here for all $i \in N$, $\mathcal{M}^C_i = v_i$ and by definition of $v$, $v_i \geq 0$.

For $\lvert C \rvert > 2$, if $i = k$, $\mathcal{M}^C_i = v_C \geq 0$ or if $i \neq k$, $\mathcal{M}^C_i = \mathcal{M}^C_k - \mathcal{M}^{C \setminus \{i\}}_k + \mathcal{M}^{C \setminus \{k\}}_i \geq v_C - v_{C \setminus \{i\}}$. This is because $\mathcal{M}^{C \setminus \{i\}}_k \leq v_{C \setminus \{i\}}$ due to R2 and $\mathcal{M}^{C \setminus \{k\}}_i \geq 0$ due to the inductive hypothesis. Finally, $v_C - v_{C \setminus \{i\}} \geq 0$ due to the properties of $v$. Therefore, the algorithm fulfills R1 for $C$.\\

R3 holds due to line 18, or, in the case of $\lvert C \rvert = 1$, due to line 4.\\

R4 is also shown by induction over the size of $C$. For any $C \in N$ with $\lvert C \rvert \leq 1$ it is easy to see that R4 must hold, since here for the single player $i \in C$, $\mathcal{M}^C_i = v_i$ holds.

For $\lvert C \rvert > 2$, we have $\mathcal{M}^C_k = v_C \geq v_k$. For $i \neq k$, $\mathcal{M}^C_i = \mathcal{M}^C_k - \mathcal{M}^{C \setminus \{i\}}_k + \mathcal{M}^{C \setminus \{k\}}_i \geq v_C - v_{C \setminus \{i\}} + v_i$. This is because $\mathcal{M}^{C \setminus \{i\}}_k \leq v_{C \setminus \{i\}}$ due to R2 and $\mathcal{M}^{C \setminus \{k\}}_i \geq v_i$ due to the inductive hypothesis. Finally, $v_C - v_{C \setminus \{i\}} + v_i \geq v_i$ due to the properties of $v$. Therefore, the algorithm fulfills R4 for $C$.\\

R5 holds simply due to line 10 or line 4, respectively.\\

F1 is shown by induction over the size of $C$.  For any $C \in N$ with $\lvert C \rvert \leq 1$ it is easy to see that F1 must hold, since here, if for some $u \in N$, $v_u =0$, $\mathcal{M}^C_u = 0$ holds. For $C = \{i\}$, with $i$ not being a ``useless'' player, when adding ``useless'' player $u \in N$, $\mathcal{M}^{C \cup \{u\}}_u = 0$ must hold, because otherwise, to fulfill F5, $i$ would have to get a reward larger than the coalition value.

For $\lvert C \rvert > 2$, by the same argument, $u$ must not be the maximum player, because otherwise other players would get rewards larger than the coalition value. Thus, the player $k \in C$ chosen in line 17 must be different from $u$.

Then, one can show that $\mathcal{M}^{C \setminus \{u\}}_k = \mathcal{M}^C_k$. 
This leads to 
\begin{align}
    \mathcal{M}^C_u = \mathcal{M}^C_k - \mathcal{M}^{C \setminus \{u\}}_k + \mathcal{M}^{C \setminus \{k\}}_u = v_C - v_C + 0 = 0    
\end{align}
and for all $i \in C \setminus \{u\}$: 
  $\mathcal{M}^C_i = $ $ \mathcal{M}^C_u - \mathcal{M}^{C \setminus \{i\}}_u + \mathcal{M}^{C \setminus \{u\}}_i
        = $ $ 0 - 0 + \mathcal{M}^{C \setminus \{u\}}_i =$ $ \mathcal{M}^{C \setminus \{u\}}_i$,
therefore F1 is fulfilled.\\


F2 is again shown by induction on the size of $C$.  For any $C \in N$ with $\lvert C \rvert \leq 2$, it is easy to see that F2 must hold, since here it only applies if $C = \{i, j\}$ and $\forall A \subseteq N \setminus \{i, j\}: v_{A \cup \{i\}} = v_{A \cup \{j\}}$. Combining $\mathcal{M}^C_i = \mathcal{M}^C_j - \mathcal{M}^{C \setminus \{i\}}_j + \mathcal{M}^{C \setminus \{j\}}_i$ with the fact that $\mathcal{M}^{C \setminus \{i\}}_j = \mathcal{M}^{\{j\}}_j = v_j$ and $\mathcal{M}^{C \setminus \{j\}}_i = \mathcal{M}^{\{i\}}_i = v_i$ yields $\mathcal{M}^C_i = \mathcal{M}^C_j$.
 For $\lvert C \rvert > 2$, one can show using the induction hypothesis and F6 that $\mathcal{M}^C_i \neq \mathcal{M}^C_j$ leads to a contradiction.\\

F3 is once again shown by induction on the size of $C$. For any $C \in N$ with $\lvert C \rvert \leq 2$, F3 does not apply, since for it to apply, $C$ must have at least two members. Therefore, the axiom holds here.

For $\lvert C \rvert > 2$, the proof works by a case analysis:  either  $\mathcal{M}_i^{C \setminus \{j\}} < v_{C \setminus \{j\}}$ holds or  $\mathcal{M}_i^{C \setminus \{j\}} = v_{C \setminus \{j\}}$. The case that  $\mathcal{M}_i^{C \setminus \{j\}} > v_{C \setminus \{j\}}$ is impossible.

If $\mathcal{M}_i^{C \setminus \{j\}} = v_{C \setminus \{j\}}$ and $\mathcal{M}_i^{C \setminus \{j\}} > \mathcal{M}_j^{C \setminus \{i\}}$, 
\begin{align}
    \mathcal{M}^C_i > \mathcal{M}^C_i - \mathcal{M}_i^{C \setminus \{j\}} + \mathcal{M}_j^{C \setminus \{i\}}
\end{align}
directly follows. Otherwise, if $\mathcal{M}_i^{C \setminus \{j\}} = v_{C \setminus \{j\}}$, one can show that there must be some $a \in C \setminus \{i, j\}$, for which $\mathcal{M}_j^{C \setminus \{a, i\}} < \mathcal{M}_i^{C \setminus \{a, j\}}$. Using this, one can show that in this case as well, $\mathcal{M}^{C \setminus \{j\}}_a > \mathcal{M}^{C \setminus \{i\}}_a$ must hold. Furthermore, due to the inductive hypothesis, $\mathcal{M}_i^{C \setminus \{a\}} > \mathcal{M}_j^{C \setminus \{a\}}$. So, the above (in-)equations hold for this case as well and therefore Algorithm \ref{solfun} fulfills F3 for $C$.

Finally, F4 is also shown by induction over the size of $C$. Let $v$ and $v'$ denote any two arbitrary but fixed value functions over all coalitions $C \subseteq N$ and $\mathcal{M}(v)^C_i$ that are monotonic with $v_{\emptyset} = v'_{\emptyset}$. For any $C \in N$ with $\lvert C \rvert = 1$, it is easy to see that if $v'_i > v_i$, then also $\mathcal{M}(v')^C_i = v'_i > v_i = \mathcal{M}(v)^C_i$.

For $\lvert C \rvert > 2$, if $v'_C > v_C$ and $\forall D \subseteq C \setminus \{i\}: v'_D = v_D$, due to theorem \ref{unique}, we have $\mathcal{M}(v')^{D}_i = \mathcal{M}(v)^{D}_i$ for all such $D.$ Then either $i$ is directly assigned $\mathcal{M}(v')^{C}_i = v'_C > v_C \geq \mathcal{M}(v)^{C}_i$, or 
\begin{align}
    \begin{split}
    \mathcal{M}(v')^{C}_i & = v'_C - \mathcal{M}(v')^{C \setminus \{i\}}_k + \mathcal{M}(v')^{C \setminus \{k\}}_i\\
    & > v_C - \mathcal{M}(v)^{C \setminus \{i\}}_k + \mathcal{M}(v)^{C \setminus \{k\}}_i
    = \mathcal{M}(v)^{C}_i,
    \end{split}
\end{align}
because Theorem \ref{unique}, $\mathcal{M}(v')^{C \setminus \{i\}}_k = \mathcal{M}(v)^{C \setminus \{i\}}_k$ and due to the inductive hypothesis as well as Theorem \ref{unique} $\mathcal{M}(v')^{C \setminus \{k\}}_i \geq \mathcal{M}(v)^{C \setminus \{k\}}_i$. Therefore, Algorithm \ref{solfun} fulfills F4 for $C$.\\

Thus, $\mathcal{M}(v)$ as computed by Algorithm \ref{solfun} fulfills R1 to R6 for any $v$ which is monotone with $v_{\emptyset} = 0$. \hfill
\textbf{End of sketch of proof}\\

To provide some intuition about our algorithm, we give a small example.
\begin{example}
Let $v$ be defined as follows:
\begin{center}
\tabcolsep=0.1cm
\begin{tabular}{|c|c|c|c|c|c|c|c|c|c|}
\hline
 $C$ & $\emptyset$ & $\{1\}$ & $\{2\}$ & $\{3\}$ & $\{4\}$ & $\{1,2\}$ & $\{1,3\}$ & $\{1,4\}$ & $\{2,3\}$\\ 
 \hline
 $v_C$ & $0$ & $1$ & $2$ & $1$ & $4$ & $3$ & $4$ & $7$ & $4$\\
 \hline
\end{tabular}\\

\begin{tabular}{|c|c|c|c|c|c|c|}
\hline
 $C$ & $\{2,4\}$ & $\{3, 4\}$ & $\{1, 2, 3\}$ & $\{1, 2, 4\}$ & $\{1, 3, 4\}$ & $\{2, 3, 4\}$\\ 
 \hline
 $v_C$ & $7$ & $6$ & $6$ & $7$ & $9$ & $9$\\
 \hline
\end{tabular}\\

\begin{tabular}{|c|c|}
\hline
 $C$ & $\{1, 2, 3, 4\}$\\ 
 \hline
 $v_C$ & $9$\\
 \hline
\end{tabular}
\end{center}

For empty or single-member coalitions, each player $i$ will just get assigned their standalone value. For larger coalitions,  the rewards are distributed as follows:

\begin{center}
\tabcolsep=0.1cm
\begin{tabular}{|c|c|c|c|c|c|c|}
\hline
 $C$ & $\{1,2\}$ & $\{1,3\}$ & $\{1, 4\}$ & $\{2,3\}$ & $\{2,4\}$ & $\{3, 4\}$\\ 
 \hline
 $\mathcal{M}(v)^C_1$ & $ 2 $ & $4$ & $4$ & $1$ & $1$ & $1$ \\
 \hline
 $\mathcal{M}(v)^C_2$ & $ 3 $ & $2$ & $2$ & $4$ & $5$ & $2$ \\
 \hline
 $\mathcal{M}(v)^C_3$ & $ 1 $ & $4$ & $1$ & $3$ & $1$ & $3$ \\
 \hline
 $\mathcal{M}(v)^C_4$ & $ 4 $ & $4$ & $7$ & $4$ & $7$ & $6$ \\
 \hline
\end{tabular}\\

\begin{tabular}{|c|c|c|c|c|c|}
\hline
 $C$ & $\{1, 2, 3\}$ & $\{1, 2, 4\}$ & $\{1, 3, 4\}$ & $\{2, 3, 4\}$ & $\{1, 2, 3, 4\}$\\ 
 \hline
 $\mathcal{M}(v)^C_1$ & $ 5 $ & $2$ & $7$ & $1$ & $5$ \\
 \hline
 $\mathcal{M}(v)^C_2$ & $ 5 $ & $3$ & $2$ & $7$ & $5$ \\
 \hline
 $\mathcal{M}(v)^C_3$ & $ 6 $ & $1$ & $6$ & $5$ & $8$  \\
 \hline
 $\mathcal{M}(v)^C_4$ & $ 4 $ & $7$ & $9$ & $9$ & $9$ \\
 \hline
\end{tabular}
\end{center}

Note how in $\{2,3,4\}$ player $2$ loses value if player $1$ joins. This is because, if player $2$ threatens to leave the coalition, $\{1,3,4\}$ will still retain a value of $9$ as compared to the value of $6$, which $\{3, 4\}$ would retain. Note further, how this example demonstrates how our F5 is incompatible with F3 as formulated by Sim and colleagues \cite{sim2020collaborative}. We have $v_2 > v_1$, but for all $C \setminus \{i, j\}$ with $C \neq \emptyset$, $v_{C \cup \{i\}} = v_{C \cup \{j\}}$. Still, in the grand coalition, $i$ and $j$ get equal rewards.

Finally, if one were to introduce a new function $v'$ with $v'(\{1, 3\}) = 5$ and otherwise $v'$ being equal to $v$, in the coalition $\{1, 3\}$, both player $1$ and player $3$ would get a larger reward. However, for the grand coalition we would get:

\begin{center}
\tabcolsep=0.1cm
\begin{tabular}{|c|c|}
\hline
 $C$ & $\{1, 2, 3, 4\}$\\ 
 \hline
 $\mathcal{M}(v')^C_1$ & $5$ \\
 \hline
 $\mathcal{M}(v')^C_2$ & $4$ \\
 \hline
 $\mathcal{M}(v')^C_3$ & $8$  \\
 \hline
 $\mathcal{M}(v')^C_4$ & $9$ \\
 \hline
\end{tabular}
\end{center}

Thus, player $1$ and player $3$ do not get an increase to their rewards, instead the reward of player $2$ is decreased. Therefore,  F4 as formulated by Sim and colleagues \cite{sim2020collaborative} is incompatible with our F5.
\end{example}

There are two relevant corollaries of the theorem. The first states that the alleged indeterminism of the Algorithm regarding which $j$ to choose (in line 12) is, in fact, none: the function calculated by the algorithm is independent of the chosen $j$.     
The second corollary expresses the fact that the Balanced Reciprocity Axiom together with the incentive axioms of weak efficiency and non-participation already entails the other axioms.  

\begin{corollary}
    For any monotone $v$, the solution $\mathcal{M}(v)$ obtained by Algorithm \ref{solfun} is identical, regardless of the player $j$ chosen in line 12.
\end{corollary}
\begin{proof}
    Let $\mathcal{M}(v)$ be obtained by Algorithm \ref{solfun}. Due to Theorem \ref{thm:r1to6}, it fulfills R3, R5 and F5, regardless of how $j$ is chosen in line 12. Therefore, it is unique and any possible choice of $j$ has to result in the same solution.
\end{proof}

\begin{corollary}
    For any monotone $v$, R1, R2, R4 and F1 to F4 are entailed by R3, R5 and F5. 
\end{corollary}
\begin{proof}
    Let $v$ be monotonic and $\mathcal{M}(v)$ satisfy R3, R5 and F6. Due to Theorem \ref{unique} it is the unique solution that satisfies these. Due to Theorem \ref{thm:r1to6}, algorithm \ref{solfun} finds this unique solution. Furthermore, also due to Theorem \ref{thm:r1to6} $\mathcal{M}(v)$ fulfills R1 to R5 and F1 to F5.
    
    Therefore, for any monotonic $v$, any solution satisfying R3, R5 and F5 must satisfy R1 to R5 and F1 to F5.
\end{proof}


\section{Scaled Shapley Falsifies Balanced Reciprocity}

We will now compare our solution function with the Shapley value. Sim and colleagues \cite{sim2020collaborative} presented a mechanism for a setting with replicable goods, similar to ours, and a solution function based on the Shapley value. As we will demonstrate, that solution function is not compatible with our Balanced Reciprocity Axiom.

The Shapley value \cite{shapley1953value} is defined as follows:
\begin{align}
    \phi_v(i) = \frac{1}{n!}\sum_{\pi \in \Pi_N}\left(v\left( S_{\pi, i}\cup \{i\} \right) - v \left(S_{\pi, i}\right) \right),
\end{align}
where $\Pi_N$ is the set of all possible permutations of $N$ and $S_{\Pi, i}$ is the coalition of parties preceding $i$ in permutation $\pi$.

However, the distribution of rewards according to the Shapley value would not fulfill R3 since no agent would get the maximum reward. Hence, Shapley rewards should be scaled in such a way that the agent with the largest contribution gets the maximum reward. How this can be done for any monotonic value function $v$ with $v_{\emptyset}\ = 0$ was shown by Sim and colleagues  \cite{sim2020collaborative} who introduced a scaled $\rho$-Shapley value.

To stay in line with our definition of a solution function from Equation \ref{defsolfun}, we adapt the scaled $\rho$-Shapley value for any coalition $C \subseteq N$, distributing a value of $v_i$ to all players $i \notin C$ outside a coalition $C \subseteq N$:
\begin{align}
    r(v)^C_i = \begin{cases}
    \left(\frac{\phi_i}{\phi^{\ast}_C} \right)^{\rho} \times v_C& \text{if $i \in C,$}\\
    v_i & \text{else},\\
    \end{cases}
\end{align}
where $\phi^{\ast}_C = \max_{i \in C}\phi_i$ and $0<\rho<1$. 

Sim and his colleagues \cite{sim2020collaborative} showed that for an appropriately chosen $\rho$, $r(v)^C_i$  fulfills R1 to R4, as well as F1 to F4. It is easy to see that for all $C \subseteq N$ and $i \in C$, $r(v)^C_i = v_i$, thus the scaled $\rho$-Shapley value fulfills R5 as well.

However, $r(v)^C_i$ is not compatible with the new Balanced Reciprocity Axiom. In particular, this result shows that our new axiom F5 is not entailed by the other axioms.  

\begin{theorem}
    The scaled $\rho$-Shapley value $r(v)^C_i$ does not fulfill F5.
\end{theorem}

\begin{proof}
We will prove this using a simple example. Let $v$ be defined as follows:
\begin{center}
\tabcolsep=0.1cm
\begin{tabular}{|c|c|c|c|c|c|c|c|c|c|}
\hline
 $C$ & $\emptyset$ & $\{1\}$ & $\{2\}$ & $\{3\}$ & $\{1,2\}$ & $\{1,3\}$ & $\{2,3\}$ & $\{1,2,3\}$\\ 
 \hline
 $v_C$ & $0$ & $1$ & $2$ & $3$ & $ 3 $ & $4$ & $5$ & $6$\\
 \hline
\end{tabular}
\end{center}

It is easy to see, that the scaled $\rho$-Shapley values and their values $\mathcal{M}(v)^C_i$ differ:
For $\rho = 1$ the scaled $\rho$-Shapley values for coalitions containing at least two members are:
\begin{center}
\begin{tabular}{|c|c|c|c|c|}
\hline
 $C$ & $\{1,2\}$ & $\{1,3\}$ & $\{2,3\}$ & $\{1,2,3\}$ \\ 
 \hline
 $r(v)^C_1$ & $ \frac{3}{2} $ & $\frac{4}{3}$ & $1$ & $2$ \\
 \hline
  $r(v)^C_2$ & $ 3 $ & $2$ & $\frac{10}{3}$ & $4$ \\
 \hline
  $r(v)^C_3$ & $ 3 $ & $4$ & $5$ & $6$ \\
 \hline
\end{tabular}
\end{center}

However, their values $\mathcal{M}(v)^C_i$ as calculated by algorithm \ref{solfun} are: 
\begin{center}
\begin{tabular}{|c|c|c|c|c|}
\hline
 $C$ & $\{1,2\}$ & $\{1,3\}$ & $\{2,3\}$ & $\{1,2,3\}$ \\ 
 \hline
 $\mathcal{M}(v)^C_1$ & $ 2 $ & $2$ & $1$ & $3$ \\
 \hline
  $\mathcal{M}(v)^C_2$ & $ 3 $ & $2$ & $4$ & $5$ \\
 \hline
  $\mathcal{M}(v)^C_3$ & $ 3 $ & $4$ & $5$ & $6$ \\
 \hline
\end{tabular}
\end{center}

Note, that $\mathcal{M}(v)^{\{1, 2\}}_1 = 2$. We can find a fitting $\rho$, so that $r(v)^{\{1, 2\}}_1$ is equal.
\begin{align}
    r(v)^{\{1, 2\}}_1 = \left(\frac{\phi_i}{\phi^{\ast}_C} \right)^{\rho} \times v_C = \left(\frac{1}{2} \right)^{\rho} \times 3
\end{align}
using $\rho = \log_2(3) - 1$, we get
\begin{align}
    r(v)^{\{1, 2\}}_1 = \left(\frac{1}{2} \right)^{\log_2(3) - 1} \times 3 = 2 = \mathcal{M}(v)^{\{1, 2\}}_1 = 2.
\end{align}

However, this $\rho$ will not be sufficient to completely align all $r(v)^C_i$ with all $\mathcal{M}(v)^C_i$. When using $\rho = \log_2(3) - 1$ for $r(v)^{\{1, 3\}}_1$, we get
\begin{align}
    r(v)^{\{1, 3\}}_1 = \left(\frac{\phi_i}{\phi^{\ast}_C} \right)^{\log_2(3) - 1} \times v_C = \left(\frac{1}{3} \right)^{\rho} \times 4 \approx 2.1036,
\end{align}
which is different from $\mathcal{M}(v)^{\{1, 3\}}_1 = 2$.
\end{proof}

Inspecting the proof, one may observe that the scaling factor $\rho$ can be used to locally achieve balanced reciprocity between two players.  However, a single $\rho$ can not be sufficient to achieve global balanced reciprocity even in this simple example.


\section{Related Work}\label{rw}

In this paper, we contribute to the game-theoretical foundations of cooperative AI \cite{conitzer24socialChoice}, with basic motivation stemming from the problem of collective data sharing for collaborative learning \cite{karimireddy2022mechanisms, sim2020collaborative, xu2021gradient, filter2024mechanisms}. A central challenge in this domain is to develop fair and effective reward mechanisms that evaluate the data contributions of each participant, ensuring adequate recognition, incentivizing participation, and maintaining fairness. To address this, numerous reward mechanisms inspired by cooperative game theory have been explored. One of them is the Shapley value \cite{lehrer1988axiomatization}, which we discussed before in this paper as a counterexample to the axiom of balanced reciprocity. Another well-known related reward function is the Banzhaf value \cite{banzhaf65weightedVoting}. Both reward mechanisms are linear with respect to the canonical addition and scalar multiplication of games and hence allow a uniform treatment in a very general setting based on concepts from linear algebra \cite{grabisch16bases}. Such generality is not easily achievable with our reward mechanism as it is not linear. This point also applies to the scaled Shapley value of Sim and colleagues and the reward mechanism of Zou and colleagues \cite{zou2020sharing} (see below).

Zou and colleagues propose a surplus allocation method for super-additive games with non-replicable goods, the proportional surplus division value \cite{zou2020sharing}. In their framework, players possess the power to act as separators, being capable of sabotaging cooperation and reducing the game to its additive core, where each participant receives only their stand-alone value. Their proposed solution distributes the surplus so that for any pair of players, the potential loss each would incur from the other defecting as a separator is proportionally equal to their standalone values. This is enforced through their axiom of "proportional loss under separatorization." This axiom resembles our balanced reciprocity axiom, but, in contrast, it contains an additional normalization factor. Moreover, we stress again that we do not restrict us to super-additive games and we deal with replicable goods.

The even stronger property of additivity of games is usually not investigated in the cooperative-game community as additive games are completely determined by the values of the individuals. However, the generalized forms of additive games (so-called k-additive games) were discussed in the context of non-balanced coalitional games \cite{gonzalez15preserving}.

The solution functions in classical cooperative game theory (in particular, the Shapley value and the Banzhaf value) assume that rewards are scarce and divisible resources, which may not align with the nature of machine learning models that can be replicated at no cost. We already referred to the work of Sim and colleagues \cite{sim2020collaborative} who address this discrepancy by adapting the Shapley value framework to settings where rewards are replicable goods, specifically ML models. In the resulting $\rho$-parameterized Shapley reward scheme, each participant receives an individualized model, where the quality of the model reflects their contribution to the collaborative training process. The authors also show that welfare maximization and fairness (based on the Shapley value) are at odds with each other. A nice feature of the  parameterization is that it allows for a fine-grained picture of the relation between welfare maximization and fairness.


\section{Conclusion}

We proposed a novel mechanism for fair reward allocation in the general setting of monotonic cooperative
games. Extending recent work on fair reward allocation for arbitrarily replicable resources, we were able to show that a reward mechanism exists that fulfills mutual fairness as explicated by the axiom of  Balanced Reciprocity Axiom. The presented algorithm  calculates  a function that is neither captured by the classical Shapley value nor any  other classical solution functions such as Banzhaf (due to their linearity).  But it also differs from recently proposed mechanisms such as the scaled Shapley value \cite{sim2020collaborative}. In fact,  as the uniqueness result shows, the provided mechanism is the only one that fulfills mutual fairness. 

For the results, we assumed that the games are monotonic, which holds for numerous practical scenarios. Interestingly, for these games the Balanced Reciprocity Axiom is that strong that it already (together with two axioms) entails the other incentivization and fairness axioms. So, one might speculate whether the Axiom of Balanced Reciprocity can be relaxed to regain independence of the axioms even for monotonic games.

Another question for future work concerns the aspect of multiple coalitions \cite{gonzalez16multicoalitional}. Classical cooperative games usually assume super-additivity  such that the formation of a grand coalition is ensured. Hence, the reward function has to be defined only for the grand coalition. We had to define the reward function for arbitrary coalitions as we did not assume super-additivity. An interesting question is what kind of coalition structures we allow: whether multiple coalitions are allowed and, if so,  whether they must be disjoint.   
Related to this question is the dynamics of coalition formation, where players enter and
exit over time.



\appendix

\section{Proof of Theorem \ref{thm:r1to6}}\label{proof:axioms}

\subsection{Proof of F5}
\begin{tabularx}{\linewidth}{l X}
    F5 & \textbf{Balanced Reciprocity:} For all $C \subseteq N$ and all $i, j \in C$:
    {\begin{align}
        \mathcal{M}^C_i - \mathcal{M}^{C \setminus \{j\}}_i = \mathcal{M}^C_j - \mathcal{M}^{C \setminus \{i\}}_j
    \end{align}}
\end{tabularx}

Let $v$ be an arbitrary but fixed function that is monotonic with $v(\emptyset) = 0$. We will prove the that algorithm \ref{solfun} fulfills F2 by induction over the size of $C$.\\

\textbf{Base case:} Let $C \subseteq N$ with $\lvert C \rvert \leq 1$. Here, F5 does not yet apply, since it only takes into account coalitions with at least two members.\\

\textbf{Inductive hypothesis:} Suppose, algorithm \ref{solfun} fulfills F5 for all coalitions $C \subseteq N$ with $\lvert C \rvert < s$ up to some size $s$, $s \geq 1$. For any $i, j \in C$ with $i \neq j$, if $i = k$ or $j = k$ with $k$ being the index picked in line 17, the axiom holds simply due to line 20 of the algorithm (This also concludes the proof for $\lvert C \rvert = 2$).

Assume now, $i \neq k \neq j$. Then $\lvert C \rvert \geq 3$ and, due line 20, the following two statements must hold:
\begin{align}
    \mathcal{M}^C_i = \mathcal{M}^C_k - \mathcal{M}^{C \setminus \{i\}}_k + \mathcal{M}^{C \setminus \{k\}}_i,
\end{align}
\begin{align}\label{eq:help5}
    \mathcal{M}^C_j = \mathcal{M}^C_k - \mathcal{M}^{C \setminus \{j\}}_k + \mathcal{M}^{C \setminus \{k\}}_j,
\end{align}
We will additionally make use of the following equations, which hold, because due to the inductive hypothesis, with $i, j, k \in C$, the algorithm fulfills F5 for $C \setminus \{i\}$, $C \setminus \{j\}$ and $C \setminus \{k\}$.:
    \begin{align}\label{eq:help1}
        \mathcal{M}^{C \setminus \{i\}}_k - \mathcal{M}^{C \setminus \{i, j\}}_k = \mathcal{M}^{C \setminus \{i\}}_j - \mathcal{M}^{C \setminus \{i, k\}}_j.
    \end{align}
    \begin{align}\label{eq:help2}
        \mathcal{M}^{C \setminus \{j\}}_i - \mathcal{M}^{C \setminus \{j, k\}}_i = \mathcal{M}^{C \setminus \{j\}}_k - \mathcal{M}^{C \setminus \{i, j\}}_k.
    \end{align}
    \begin{align}\label{eq:help3}
        \mathcal{M}^{C \setminus \{k\}}_i - \mathcal{M}^{C \setminus \{j, k\}}_i = \mathcal{M}^{C \setminus \{k\}}_j - \mathcal{M}^{C \setminus \{i, k\}}_j,
    \end{align}

The following chain of equations holds:
\begin{align}
    \mathcal{M}^C_i = \mathcal{M}^C_k - \mathcal{M}^{C \setminus \{i\}}_k + \mathcal{M}^{C \setminus \{k\}}_i
\end{align}
Substituting $\mathcal{M}^{C \setminus \{i\}}_k$ according to equation \ref{eq:help1} yields:
\begin{align}
    \begin{split}
        & \mathcal{M}^C_i = \mathcal{M}^C_k + \mathcal{M}^{C \setminus \{k\}}_i\\
        & - \left( \mathcal{M}^{C \setminus \{i\}}_j + \mathcal{M}^{C \setminus \{i, j\}}_k - \mathcal{M}^{C \setminus \{i, k\}}_j \right)\\
        \Rightarrow & \mathcal{M}^C_i = \mathcal{M}^C_k + \mathcal{M}^{C \setminus \{k\}}_i +\mathcal{M}^{C \setminus \{i, k\}}_j -\mathcal{M}^{C \setminus \{i\}}_j - \mathcal{M}^{C \setminus \{i, j\}}_k\\
    \end{split}
\end{align}
Substituting $\mathcal{M}^{C \setminus \{i, j\}}_k$ according to equation \ref{eq:help2} yields:
\begin{align}
    \begin{split}
        & \mathcal{M}^C_i = \mathcal{M}^C_k + \mathcal{M}^{C \setminus \{k\}}_i +\mathcal{M}^{C \setminus \{i, k\}}_j -\mathcal{M}^{C \setminus \{i\}}_j\\
        &- \left( \mathcal{M}^{C \setminus \{j\}}_k + \mathcal{M}^{C \setminus \{j, k\}}_i - \mathcal{M}^{C \setminus \{j\}}_i \right)\\
        \Rightarrow & \mathcal{M}^C_i - \mathcal{M}^{C \setminus \{j\}}_i  = \mathcal{M}^C_k + \mathcal{M}^{C \setminus \{k\}}_i + \mathcal{M}^{C \setminus \{i, k\}}_j -\mathcal{M}^{C \setminus \{i\}}_j\\
        &- \mathcal{M}^{C \setminus \{j\}}_k - \mathcal{M}^{C \setminus \{j, k\}}_i
    \end{split}
\end{align}
Substituting $\mathcal{M}^C_k$ according to equation \ref{eq:help5} yields:
\begin{align}
    \begin{split}
        & \mathcal{M}^C_i - \mathcal{M}^{C \setminus \{j\}}_i  = \left( \mathcal{M}^C_j + \mathcal{M}^{C \setminus \{j\}}_k - \mathcal{M}^{C \setminus \{k\}}_j \right)\\
        & + \mathcal{M}^{C \setminus \{k\}}_i + \mathcal{M}^{C \setminus \{i, k\}}_j -\mathcal{M}^{C \setminus \{i\}}_j - \mathcal{M}^{C \setminus \{j\}}_k - \mathcal{M}^{C \setminus \{j, k\}}_i\\
        \Rightarrow & \mathcal{M}^C_i - \mathcal{M}^{C \setminus \{j\}}_i  = \mathcal{M}^C_j - \mathcal{M}^{C \setminus \{i\}}_j\\
        & + \mathcal{M}^{C \setminus \{k\}}_i - \mathcal{M}^{C \setminus \{j, k\}}_i -\mathcal{M}^{C \setminus \{k\}}_j + \mathcal{M}^{C \setminus \{i, k\}}_j\\
    \end{split}
\end{align}
Here, due to equation \ref{eq:help3} the last four terms cancel out, thus
\begin{align}
    \mathcal{M}^C_i - \mathcal{M}^{C \setminus \{j\}}_i  = \mathcal{M}^C_j - \mathcal{M}^{C \setminus \{i\}}_j
\end{align}
Therefore, for all pairs $i, j \in C$ with $i \neq j$, we have $\mathcal{M}^C_i - \mathcal{M}^{C \setminus \{j\}}_i = \mathcal{M}^C_j - \mathcal{M}^{C \setminus \{i\}}_j$ and therefore the algorithm fulfills F5 for $C$. By the principle of mathematical induction, algorithm \ref{solfun} fulfills F5 for all $C \subseteq N$ with $\lvert C \rvert = s$ and $s \leq |N|$.


\subsection{Proof of R2}
\begin{tabularx}{\linewidth}{l X}
        R2 &\textbf{Feasibility:} $\forall C \subseteq N$, $\forall i \in C: \mathcal{M}^C_i \leq v_C$.
\end{tabularx}

Let $v$ be an arbitrary but fixed function that is monotonic with $v(\emptyset) = 0$, let $C \subseteq N$. For $C = \emptyset$, this holds trivially, since no $i \in C$ exists.

If $\lvert C \rvert \leq 1$, let $i$ be the one player in $C$. Then $v_C = v_i$ due to line 4 and therefore $\mathcal{M}_i^C = v_i = v_C$.

Now, suppose $\lvert C \rvert \geq 2$. We first note that if $i$ is index $k$ picked in line 17, due to line 18  $\mathcal{M}_i^C = v_C$, thus the axiom is fulfilled here.

If $i = j \neq k$, that is $i$ was picked as $j$ in line 12, but then some other $k \in C$, $k \neq i$ was chosen in line 17, we have $m_i^C \leq m_k^C$, otherwise $i$ would have been picked in line 17. Since $i = j$, $k$ was assigned $m^C_k = m^C_i - \mathcal{M}^{C \setminus \{k\}}_i + \mathcal{M}^{C \setminus \{i\}}_k$ in line 15. Thus, 

\begin{align}
    \begin{split}
         & m_i^C  \leq m^C_i - \mathcal{M}^{C \setminus \{k\}}_i + \mathcal{M}^{C \setminus \{i\}}_k = m_k^C,
    \end{split}
\end{align}

so $\mathcal{M}^{C \setminus \{k\}}_i \leq \mathcal{M}^{C \setminus \{i\}}_k$. $\mathcal{M}_i^C$ is then assigned in line 20: $\mathcal{M}^C_i = \mathcal{M}^C_k - \mathcal{M}^{C \setminus \{i\}}_k + \mathcal{M}^{C \setminus \{k\}}_i$. With $\mathcal{M}_k^C = v_C$ from line 18, this yields $\mathcal{M}_i^C \leq \mathcal{M}_k^C = v_C$.

Lastly, if $j \neq i \neq k$, since $k$ was selected in line 17, $m^C_k \geq m_i^C$ and thus from line 15
\begin{align}
    \begin{split}
        & m^C_k = m^C_j - \mathcal{M}^{C \setminus \{k\}}_j + \mathcal{M}^{C \setminus \{j\}}_k\\
        \geq & m^C_j - \mathcal{M}^{C \setminus \{i\}}_j + \mathcal{M}^{C \setminus \{j\}}_i = m^C_i.
    \end{split}
\end{align}

Therefore, $\mathcal{M}^{C \setminus \{j\}}_k - \mathcal{M}^{C \setminus \{k\}}_j \geq \mathcal{M}^{C \setminus \{j\}}_i - \mathcal{M}^{C \setminus \{i\}}_j$. By using the three equations \ref{eq:help1}, \ref{eq:help2} and \ref{eq:help3} from the proof of F5, which hold, because the algorithm fulfills R5, we can show the following:
\begin{align}
    \begin{split}
        & \mathcal{M}^{C \setminus \{j\}}_k - \mathcal{M}^{C \setminus \{k\}}_j \geq \mathcal{M}^{C \setminus \{j\}}_i - \mathcal{M}^{C \setminus \{i\}}_j
    \end{split}
\end{align}
Substituting $\mathcal{M}^{C \setminus \{k\}}_j$ according to equation \ref{eq:help3} and $\mathcal{M}^{C \setminus \{j\}}_j$ according to equation \ref{eq:help2} yields:
\begin{align}
    \begin{split}
        \Rightarrow & \mathcal{M}^{C \setminus \{j\}}_k - \left( \mathcal{M}^{C \setminus \{k\}}_i + \mathcal{M}^{C \setminus \{i, k\}}_j - \mathcal{M}^{C \setminus \{j, k\}}_i \right)\\
        & \geq \mathcal{M}^{C \setminus \{j\}}_i - \left( \mathcal{M}^{C \setminus \{i\}}_k + \mathcal{M}^{C \setminus \{i, k\}}_j - \mathcal{M}^{C \setminus \{i, j\}}_k \right)\\
        \Rightarrow & \mathcal{M}^{C \setminus \{j\}}_k - \mathcal{M}^{C \setminus \{k\}}_i + \mathcal{M}^{C \setminus \{j, k\}}_i\\
        & \geq \mathcal{M}^{C \setminus \{j\}}_i - \mathcal{M}^{C \setminus \{i\}}_k + \mathcal{M}^{C \setminus \{i, j\}}_k
    \end{split}
\end{align}
Substituting $\mathcal{M}^{C \setminus \{i, j\}}_k$ according to equation \ref{eq:help2} yields:
\begin{align}
    \begin{split}
            \Rightarrow & \mathcal{M}^{C \setminus \{j\}}_k - \mathcal{M}^{C \setminus \{k\}}_i + \mathcal{M}^{C \setminus \{j, k\}}_i\\
            & \geq \mathcal{M}^{C \setminus \{j\}}_i - \mathcal{M}^{C \setminus \{i\}}_k + \left( \mathcal{M}^{C \setminus \{j\}}_k + \mathcal{M}^{C \setminus \{j, k\}}_ij- \mathcal{M}^{C \setminus \{j\}}_i \right)\\
            \Rightarrow & \mathcal{M}^{C \setminus \{j\}}_k - \mathcal{M}^{C \setminus \{k\}}_i + \mathcal{M}^{C \setminus \{j, k\}}_i\\
            & \geq - \mathcal{M}^{C \setminus \{i\}}_k + \mathcal{M}^{C \setminus \{j\}}_k + \mathcal{M}^{C \setminus \{j, k\}}_i\\
            \Rightarrow & - \mathcal{M}^{C \setminus \{k\}}_i \geq - \mathcal{M}^{C \setminus \{i\}}_k\\
            \Rightarrow & \mathcal{M}^{C \setminus \{k\}}_i \leq \mathcal{M}^{C \setminus \{i\}}_k\\
        \end{split}
    \end{align}
Thus, for all $i \in C$ with $i\neq k$, $\mathcal{M}^{C \setminus \{k\}}_i \leq \mathcal{M}^{C \setminus \{i\}}_k$ and therefore $\mathcal{M}^C_i = \mathcal{M}^C_k - \mathcal{M}^{C \setminus \{i\}}_k + \mathcal{M}^{C \setminus \{k\}}_i \leq \mathcal{M}^C_k = v_C$. Thus, the algorithm fulfills R2 for $C$.


\subsection{Proof of R1}
\begin{tabularx}{\linewidth}{l X}
    R1 & \textbf{Non-negativity:} $\forall C \subseteq N$, $\forall i \in C: \mathcal{M}^C_i \geq 0$.
\end{tabularx}\\

Let $v$ be an arbitrary but fixed function that is monotonic with $v(\emptyset) = 0$. We prove the algorithm \ref{solfun} fulfills R1 by induction over the size of $C$.\\

\textbf{Base case:} Let $C \subseteq N$ with $\lvert C \rvert \leq 1$. It is easy to see that R1 holds here. Each player $i \in C$ is assigned their standalone value $v_i$ in line 4. Since $v$ is monotonic with $v(\emptyset) = 0$, we  $v_i \geq 0$ for all $i \in N$, thus $\mathcal{M}^C_i \geq 0$\\

\textbf{Inductive hypothesis:} Suppose, algorithm \ref{solfun} fulfills R1 for all coalitions $C \subseteq N$ with $\lvert C \rvert < s$ up to some size $s$, $s \geq 1$.\\

\textbf{Inductive step:} Let $C \subseteq N$ with $\lvert C \rvert=s$. We will show that R1 will hold here as well. For any $i \in C$, if $i$ is the index $k$ picked in line 17, then $\mathcal{M}^C_i = v_C$ and because $v$ is monotonic with $v(\emptyset) = 0$, $v_C \geq 0$ for all $C \subseteq N$.

Otherwise, $M_i^C$ is assigned in line 20: $\mathcal{M}^C_i = \mathcal{M}^C_k - \mathcal{M}^{C \setminus \{i\}}_k +\mathcal{M}^{C \setminus \{k\}}_i \geq v_C - v_{C \setminus \{i\}}$. This is, because $\mathcal{M}^C_k = v_C$ (line 17), $\mathcal{M}^{C \setminus \{i\}}_k \leq v_{C \setminus \{i\}}$ (due to R2) and $\mathcal{M}^{C \setminus \{k\}}_i \geq 0$ (due to the inductive hypothesis the algorithm fulfills R1 for $C \setminus \{k\}$). Furthermore, due to the monotonicity property of $v$, $v_C \geq v_{C \setminus \{i\}}$, thus $\mathcal{M}^C_i  \geq v_C - v_{C \setminus \{i\}} \geq 0$. Therefore, the algorithm fulfills R1 for $C$. By the principle of mathematical induction, algorithm \ref{solfun} fulfills R1 for all $C \subseteq N$ with $\lvert C \rvert = s$ and $s \leq |N|$.


\subsection{Proof of R3}
\begin{tabularx}{\linewidth}{l X}
        R3 & \textbf{Weak Efficiency:} $\forall C \subseteq N\; \exists i \in C: \mathcal{M}^C_i = v_C$.
\end{tabularx}\\

Let $C \subseteq N$. For $C = \emptyset$ this axiom does not apply. For $\lvert C \rvert = 1$ with $C = \{i\}$, we have due to line 4 $\mathcal{M}^C_i = v_i = v_C$. If $\lvert C \rvert \geq 2$, this axioms holds due to line 18, which assigns $\mathcal{M}^C_k = v_C$ to some $k \in C$. Therefore, the algorithm fulfills R3.


\subsection{Proof of R4}
\begin{tabularx}{\linewidth}{l X}
        R4 & \textbf{Individual Rationality:} $\forall C \subseteq N$, $\forall i \in N: \mathcal{M}^C_i \geq v_i.$
\end{tabularx}\\

Let $v$ be an arbitrary but fixed function that is monotonic with $v(\emptyset) = 0$. We prove the algorithm \ref{solfun} fulfills R4 by induction over the size of $C$.\\

\textbf{Base case:} Let $C \subseteq N$ with $\lvert C \rvert \leq 1$. R4 holds, since for all $C \subseteq N$ with $\lvert C \rvert \leq 1$ and all $i \in N$, $\mathcal{M}_i^{\emptyset} = v_i$ and $\mathcal{M}_i^C = v_i$ are directly assigned in line 2 and 4.\\

\textbf{Inductive hypothesis:} Suppose, algorithm \ref{solfun} fulfills R4 for all coalitions $C \subseteq N$ with $\lvert C \rvert < s$ up to some size $s$, $s \geq 1$.\\

\textbf{Inductive step:} Let $C \subseteq N$ with $\lvert C \rvert=s$ and $i \in C$. If $i = k$, then due to line 18, $\mathcal{M}^C_i = v_C$ and due to the monotonicity property of $v$, $v_C \geq v_i$. 

If $i \neq k$, $i$ is assigned $\mathcal{M}^C_i = \mathcal{M}^C_k - \mathcal{M}^{C \setminus \{i\}}_k + \mathcal{M}^{C \setminus \{k\}}_i$ in line 20. The following holds: $\mathcal{M}^C_k = v_C$ (due to line 18), $\mathcal{M}^{C \setminus \{i\}}_k \leq v_{C \setminus \{i\}} \leq v_C$ (due to R2 and the monotonicity property of $v$) and $\mathcal{M}^{C \setminus \{k\}}_i \geq v_i$ (due to the inductive hypothesis, the algorithm fulfills R4 for $C \setminus \{k\}$). Thus $\mathcal{M}^C_i = \mathcal{M}^C_k - \mathcal{M}^{C \setminus \{i\}}_k + \mathcal{M}^{C \setminus \{k\}}_i \geq v_C - v_C + v_i = v_i$.  Therefore, the algorithm fulfills R1 for $C$. By the principle of mathematical induction, algorithm \ref{solfun} fulfills R4 for all $C \subseteq N$ with $\lvert C \rvert = s$ and $s \leq |N|$.


\subsection{Proof of R5}
\begin{tabularx}{\linewidth}{l X}
        R5 & \textbf{Non-participation:} $\forall C \subseteq N$, $\forall i \notin C: \mathcal{M}^C_i = v_i$.
\end{tabularx}\\

Let $C \subseteq N$. If $C = \emptyset$, this axiom does not apply. If $\lvert C \rvert = 1$, due to line 4, for all $i \in N$, we have $\mathcal{M}^C_i = v_i$. If $\lvert C \rvert > 1$, due to line 10, for all $i \notin C$, we have $\mathcal{M}^C_i = v_i$. Thus, algorithm \ref{solfun} fulfills R5.


\subsection{Proof of F1}
\begin{tabularx}{\linewidth}{l X}
    F1 & \textbf{Uselessness:} For all $u \in N$,
    {\begin{align}
        \begin{split}
            & \left( \forall C \subseteq N \setminus \{u\}: v_C = v_{C \cup \{u\}} \right)\\
            \Rightarrow & \big( \left(\forall C \subseteq N: \mathcal{M}^C_u = 0 \right)\\
            & \land \left( \forall C \subseteq N \setminus \{u\}, \forall i \in C: \mathcal{M}^C_i = \mathcal{M}^{C \cup \{u\}}_i \right) \big).
        \end{split}
    \end{align}}
\end{tabularx}\\

Let $v$ be an arbitrary but fixed function that is monotonic with $v(\emptyset) = 0$. We prove the algorithm \ref{solfun} fulfills F1 by induction over the size of $C$.\\

\textbf{Base case:} Let $C \subseteq N$ with $\lvert C \rvert \leq 1$. Let $u$ be "useless, so $\forall C \subseteq N \setminus \{u\}: v_C = v_{C \cup \{u\}}$. Then $v_u = 0$ must hold as well, thus for all $C \subseteq N$ with $\lvert C \rvert \leq 1$,  $\mathcal{M}_u^C = v_u = 0$ due to line 2 and line 4. 

For the second part, if $C = \emptyset$, it does not apply, since there is no player in $C$. If $C = \{i\}$ for any $i \in N$ with $i \neq u$, we have $\mathcal{M}^{C \cup \{u\}}_i = \mathcal{M}^{C \cup \{u\}}_u - \mathcal{M}^{C}_u + \mathcal{M}^{\{u\}}_i = \mathcal{M}^C_u - 0 + v_i$ has to hold ($\mathcal{M}^{C}_u = v_u = 0$ and $\mathcal{M}^{\{u\}}_i = v_i$ both due to line 4). Due to R2, we have $\mathcal{M}^C_i \leq v_{{C \cup \{u\}}} = v_C = v_i$, thus $v_i = \mathcal{M}^C_i = \mathcal{M}^{C \cup \{u\}}_i = v_i$.\\

\textbf{Inductive hypothesis:} Suppose, algorithm \ref{solfun} fulfills F1 for all coalitions $C \subseteq N$ with $\lvert C \rvert < s$ up to some size $s$, $s \geq 1$.\\

\textbf{Inductive step:} Let $C \subseteq N$ with $\lvert C \rvert=s$. Let $u$ be "useless, so $\forall C \subseteq N \setminus \{u\}: v_C = v_{C \cup \{u\}}$
. 
First we will examine the special case, that all players in $C$ are "useless". Then $v_C = v_{\emptyset} = 0$ and since we have already seen that the algorithms fulfills R1 and R2 for $C$, in this case we have for all $i \in C$: $\mathcal{M}_i^C = 0$.
    
Suppose now, there exists $u \in C$ such that $\left( \forall A \subseteq N \setminus \{u\}: v_A = v_{A \cup \{u\}} \right)$, but $v_C > 0$ (and therefore also $v_{C \setminus \{u\} > 0}$). We will first see, that  $u$ cannot be the player $k$ chosen in line 17. Due to F5, for all $i \in C \setminus \{k\}$: $\mathcal{M}^C_i -  \mathcal{M}^{C \setminus \{k\}}_i = \mathcal{M}^C_k - \mathcal{M}^{C \setminus \{i\}}_k$, with $\mathcal{M}^C_k = v_C$ and due to R2 $\mathcal{M}^C_i \leq v_C$, therefore $\mathcal{M}^{C \setminus \{k\}}_i \leq \mathcal{M}^{C \setminus \{i\}}_k$. If $k=u$, this would give $\mathcal{M}^{C \setminus \{u\}}_i \leq 0$, since $\mathcal{M}^{C \setminus \{i\}}_u = 0$ (because, due to the induction hypothesis, the algorithm fulfills F1 for $C \setminus \{i\}$). But this would violate R3, since now for all $i \in C$: $\mathcal{M}^{C \setminus \{u\}}_i \leq 0 < v_{C \setminus \{u\}}$. Thus, $k \neq u$.

Next, we will see that $\mathcal{M}^{C \setminus \{u\}}_k = \mathcal{M}^C_k$ with $k$ being the player chosen in line 17. Suppose this were not the case, so $\mathcal{M}^{C \setminus \{u\}}_k < \mathcal{M}^C_k$ ($\mathcal{M}^{C \setminus \{u\}}_k > \mathcal{M}^C_k$ is not possible, since due to R2 $\mathcal{M}^{C \setminus \{u\}}_k \leq v_{C \setminus \{u\}} = v_C = \mathcal{M}^C_k$). Then, due to R3, there must be some $k' \in C \setminus \{u, k\}$ with $\mathcal{M}_{k'}^{C \setminus \{u\}} =  v_{C \setminus \{u\}} = v_C > \mathcal{M}_k^{C \setminus \{u\}}$. Due to F6, we have 
\begin{align}
    \mathcal{M}_k^{C \setminus \{u\}} - \mathcal{M}_k^{C \setminus \{k', u\}} = \mathcal{M}_{k'}^{C \setminus \{u\}} - \mathcal{M}_{k'}^{C \setminus \{k, u\}}
\end{align}
and therefore
\begin{align}
    \mathcal{M}_k^{C \setminus \{u\}} = \mathcal{M}_{k'}^{C \setminus \{u\}} - \mathcal{M}_{k'}^{C \setminus \{k, u\}} + \mathcal{M}_k^{C \setminus \{k', u\}}    
\end{align}
If $\mathcal{M}_{k'}^{C \setminus \{u\}} > \mathcal{M}_k^{C \setminus \{u\}}$, then $\mathcal{M}_{k'}^{C \setminus \{k, u\}} > \mathcal{M}_k^{C \setminus \{k', u\}}$ would have to hold. Furthermore, since due to the induction hypothesis the algorithm fulfills F1 for $C \setminus \{k\}$ and $C \setminus \{k'\}$, $\mathcal{M}_{k'}^{C \setminus \{k, u\}} = \mathcal{M}_{k'}^{C \setminus \{k\}}$ and $\mathcal{M}_k^{C \setminus \{k', u\}} = \mathcal{M}_k^{C \setminus \{k'\}}$ and thereby $\mathcal{M}_{k'}^{C \setminus \{k\}} > \mathcal{M}_{k}^{C \setminus \{k'\}}$ would have to hold. However, since the algorithm fulfills F5, we have 
\begin{align}
    \mathcal{M}_k^{C} - \mathcal{M}_k^{C \setminus \{k'\}} = \mathcal{M}_{k'}^{C} - \mathcal{M}_{k'}^{C \setminus \{k\}}    
\end{align}
But with $\mathcal{M}_{k'}^{C \setminus \{k\}} > \mathcal{M}_{k}^{C \setminus \{k'\}}$ this would imply $\mathcal{M}_{k'}^{C} > \mathcal{M}_{k}^{C}$. But this would violate R3, since $\mathcal{M}_{k}^{C} = v_C$. Therefore, $\mathcal{M}^{C \setminus \{u\}}_k < \mathcal{M}^C_k$ is a contradiction and $\mathcal{M}^{C \setminus \{u\}}_k = \mathcal{M}^C_k$ must hold.

Using this, due to F5 we have
\begin{align}
    \mathcal{M}^C_u = \mathcal{M}^C_k - \mathcal{M}^{C \setminus \{u\}}_k + \mathcal{M}^{C \setminus \{k\}}_u = v_C - v_C + 0 = 0    
\end{align}
($\mathcal{M}^{C \setminus \{k\}}_u = 0$ because due to the induction hypothesis, the algorithm fulfills F1 for $C \setminus \{k\}$).

Furthermore, again due to F5 and using $\mathcal{M}^C_u = \mathcal{M}^{C \setminus \{i\}}_u = 0$, for all $i \in C \setminus \{u\}$ we have
\begin{align}
    \begin{split}
        \mathcal{M}^C_i = & \mathcal{M}^C_u - \mathcal{M}^{C \setminus \{i\}}_u + \mathcal{M}^{C \setminus \{u\}}_i\\
        = & 0 - 0 + \mathcal{M}^{C \setminus \{u\}}_i\\
        = & \mathcal{M}^{C \setminus \{u\}}_i,
    \end{split}
\end{align}
Thus, the algorithm fulfills axiom F1 for $C$ as well. By the principle of mathematical induction, algorithm \ref{solfun} fulfills F1 for all $C \subseteq N$ with $\lvert C \rvert = s$ and $s \leq |N|$.\\


\subsection{Proof of F2}
\begin{tabularx}{\linewidth}{l X}
    F2 & \textbf{Symmetry:} For all $i,j \in N$ s.t. $i\neq j$,
    {\begin{align}
        \begin{split}
            & \left( \forall C \subseteq N \setminus \{i, j\}: v_{C \cup \{i\}} = v_{C \cup \{j\}} \right)\\
            \Rightarrow & \left(\forall C \subseteq N \text{ with } i, j \in C: \mathcal{M}^C_i = \mathcal{M}^C_j\right).
        \end{split}
    \end{align}}
\end{tabularx}\\

Let $v$ be an arbitrary but fixed function that is monotonic with $v(\emptyset) = 0$. We prove the algorithm \ref{solfun} fulfills F2 by induction over the size of $C$.\\

\textbf{Base case:} Let $C \subseteq N$ with $\lvert C \rvert \leq 2$. If $ \exists i,j \in N$ with $i\neq j$, so that $\forall A \subseteq N \setminus \{i, j\}: v_{A \cup \{i\}} = v_{A \cup \{j\}}$, this also means $v_i = v_j$ ($A = \emptyset$). If $\lvert C \rvert < 2$ or $C \neq \{i, j\}$, this axiom does not apply here.

If $C = \{i, j\}$, we have due to F5:
\begin{align}
    \mathcal{M}^C_i = \mathcal{M}^C_j - \mathcal{M}^{C \setminus \{i\}}_j + \mathcal{M}^{C \setminus \{j\}}_i
\end{align}
Due to R3, have $\mathcal{M}^{C \setminus \{i\}}_j = \mathcal{M}^{\{j\}}_j = v_j$ and $\mathcal{M}^{C \setminus \{j\}}_i = \mathcal{M}^{\{i\}}_i = v_i$ and therfore 
\begin{align}
    \mathcal{M}^C_i = \mathcal{M}^C_j - v_j + v_i = \mathcal{M}^C_j.
\end{align}
Thus the axiom holds for all $C \subseteq N$ with $\lvert C \rvert \leq 2$.\\

\textbf{Inductive hypothesis:} Suppose, algorithm \ref{solfun} fulfills F2 for all coalitions $C \subseteq N$ with $\lvert C \rvert < s$ up to some size $s$, $s \geq 0$.\\

\textbf{Inductive step:} Let $C \subseteq N$ with $\lvert C \rvert=s$. Suppose $\exists i,j \in C$ with $i\neq j$, so that $\left( \forall A \subseteq N \setminus \{i, j\}: v_{A \cup \{i\}} = v_{A \cup \{j\}} \right)$. We will see that $\mathcal{M}^C_i = \mathcal{M}^C_j$.
    
Due to F6, we know that $\mathcal{M}^C_i - \mathcal{M}^{C \setminus \{j\}}_i = \mathcal{M}^C_j - \mathcal{M}^{C \setminus \{i\}}_j$. Suppose now $\mathcal{M}^C_i \neq \mathcal{M}^C_j$. We will see that this leads to a contradiction.

If $\mathcal{M}^C_i \neq \mathcal{M}^C_j$, then $\mathcal{M}^{C \setminus \{j\}}_i \neq \mathcal{M}^{C \setminus \{i\}}_j$ must hold as well. W.l.o.g. we will assume $\mathcal{M}^{C \setminus \{j\}}_i < \mathcal{M}^{C \setminus \{i\}}_j$. Then also $\mathcal{M}^{C \setminus \{j\}}_i < v_{C \setminus\{j\}}$ , because due to R2 $\mathcal{M}^{C \setminus \{i\}}_j \leq v_{C \setminus\{i\}} = v_{C \setminus\{j\}}$. Let $a$ be the player for which $\mathcal{M}^{C \setminus \{j\}}_a = v_{C \setminus\{j\}}$, due to R3 such a player must exist. Due to F6 we have 
\begin{align}
    \begin{split}
        & \mathcal{M}^{C \setminus \{j\}}_i = \mathcal{M}^{C \setminus \{j\}}_a - \mathcal{M}^{C \setminus \{i, j\}}_a + \mathcal{M}^{C \setminus \{a, j\}}_i\\
        < & \mathcal{M}^{C \setminus \{i\}}_j = \mathcal{M}^{C \setminus \{i\}}_a - \mathcal{M}^{C \setminus \{i, j\}}_a + \mathcal{M}^{C \setminus \{a, i\}}_j.
    \end{split}
\end{align}
Adding $\mathcal{M}^{C \setminus \{i, j\}}_a$ on both sides yields
\begin{align}
    \mathcal{M}^{C \setminus \{j\}}_a + \mathcal{M}^{C \setminus \{a, j\}}_i < \mathcal{M}^{C \setminus \{i\}}_a + \mathcal{M}^{C \setminus \{a, i\}}_j.
\end{align}
Due to $v_{C \setminus \{i\}} = v_{C \setminus \{j\}} = \mathcal{M}^{C \setminus \{j\}}_a$ and $\mathcal{M}^{C \setminus \{i\}}_a \leq v_{C \setminus \{i\}}$ (because of R2) we have $\mathcal{M}^{C \setminus \{j\}}_a \geq \mathcal{M}^{C \setminus \{i\}}_a$, which leads to
\begin{align}
    \mathcal{M}^{C \setminus \{a, j\}}_i < \mathcal{M}^{C \setminus \{a, i\}}_j.
\end{align}
However, due to F6 we have $\mathcal{M}^{C \setminus \{a\}}_i - \mathcal{M}^{C \setminus \{a, j\}}_i = \mathcal{M}^{C \setminus \{a\}}_j - \mathcal{M}^{C \setminus \{a, i\}}_j$ and, since due to the inductive hypothesis the algorithm fulfills F2 for $C \setminus \{a\}$, we have $\mathcal{M}^{C \setminus \{a\}}_i = \mathcal{M}^{C \setminus \{a\}}_j$, thus also $\mathcal{M}^{C \setminus \{a, j\}}_i = \mathcal{M}^{C \setminus \{a, i\}}_j$. Therefore, $\mathcal{M}^{C \setminus \{a, j\}}_i < \mathcal{M}^{C \setminus \{a, i\}}_j$ is a contradiction.
Thus $\mathcal{M}^{C \setminus \{j\}}_i = \mathcal{M}^{C \setminus \{i\}}_j$ has to hold and therefore also $\mathcal{M}^C_i = \mathcal{M}^C_j$. Then the algorithm fulfills F2 for $C$ and by the principle of mathematical induction, algorithm \ref{solfun} fulfills F2 for all $C \subseteq N$ with $\lvert C \rvert = s$ and $s \leq |N|$.\\


\subsection{Proof of F3}
\begin{tabularx}{\linewidth}{l X}
    F3 & \textbf{Strict Desirability:} For all $i,j \in N$ s.t. $i\neq j$, 
    {\begin{align}
        \begin{split}
            & \left( \forall A \subseteq N \setminus \{i, j\}: v_{A \cup \{i\}} \geq v_{A \cup \{j\}} \right) \\
            \Rightarrow & (\forall C \subseteq N \text{ with } i, j \in C \text{ and } \exists B \subseteq C \setminus \{i, j\}, B \neq \emptyset,\\
            & \text{ s.t. } v_{B \cup \{i\}} > v_{B \cup \{j\}} ): \left( \mathcal{M}^C_i > \mathcal{M}^C_j\right).
        \end{split}
    \end{align}}
\end{tabularx}\\

Let $v$ be an arbitrary but fixed function that is monotonic with $v(\emptyset) = 0$. We show that the algorithm \ref{solfun} fulfills F3 by induction over the size of $C$.\\

\textbf{Base case:} Let $C \subseteq N$ with $\lvert C \rvert \leq 2$. Since the axiom only applies to coalitions with at least three members, it does not apply here.\\

\textbf{Inductive hypothesis:} Suppose, algorithm \ref{solfun} fulfills F2 for all coalitions $C \subseteq N$ with $\lvert C \rvert < s$ up to some size $s$, $s \geq 2$.\\

\textbf{Inductive step:} Let $C \subseteq N$ with $\lvert C \rvert=s$. Suppose $\exists i,j \in C$ with $i\neq j$, so that $\forall A \subseteq N \setminus \{i, j\}: v_{A \cup \{i\}} \geq v_{A \cup \{j\}}$ and $\exists B \subseteq C \setminus \{i, j\}, B \neq \emptyset$, for which $v_{B \cup \{i\}} > v_{B \cup \{j\}}$. We will see that $\mathcal{M}^C_i > \mathcal{M}^C_j$.

We have to distinguish two cases, first $\mathcal{M}_i^{C \setminus \{j\}} < v_{C \setminus \{j\}}$ and second $\mathcal{M}_a^{C \setminus \{j\}} = v_{C \setminus \{j\}}$. Note, that due to R3, $\mathcal{M}_a^{C \setminus \{j\}} > v_{C \setminus \{j\}}$ is impossible.\\

\textbf{Case 1: } $\mathcal{M}_i^{C \setminus \{j\}} < v_{C \setminus \{j\}}$.

Let $a \in C \setminus \{j\}$ be the player such that $\mathcal{M}_a^{C \setminus \{j\}} = v_{C \setminus \{j\}}$. Such a player must exist due to R3 and $a \neq i$. Since $v_{C \setminus \{j\}} \geq v_{C \setminus \{i\}}$, $\mathcal{M}_a^{C \setminus \{j\}} = v_{C \setminus \{j\}} \geq v_{C \setminus \{i\}} \geq \mathcal{M}_a^{C \setminus \{i\}}$ must hold (the last inequality again due to R3).
    
First, suppose $ \exists B' \subseteq C \setminus \{i, j, a\}, B' \neq \emptyset$ with $v_{B' \cup \{i\}} > v_{B' \cup \{j\}}$. Then, because due to the inductive hypothesis the algorithm fulfills F3 for $C \setminus \{a\}$, $\mathcal{M}_i^{C \setminus \{a\}} > \mathcal{M}_j^{C \setminus \{a\}}$ must hold. Using $\mathcal{M}_a^{C \setminus \{j\}} \geq \mathcal{M}_a^{C \setminus \{i\}}$ and F5, we have:
\begin{align*}
    &\mathcal{M}_i^C=&\mathcal{M}_a^C - \mathcal{M}_a^{C \setminus \{i\}} + \mathcal{M}_i^{C \setminus \{a\}}& \\
    >&&\mathcal{M}_a^C - \mathcal{M}_a^{C \setminus \{j\}} + \mathcal{M}_j^{C \setminus \{a\}} & = \mathcal{M}_j^C.
\end{align*}
So the algorithm fulfills F3 for $C$.
    
Now, suppose no such B' exists, thus we have for all $A' \subseteq C \setminus \{i, j, a\}$: $v_{A \cup \{i\}} = v_{A \cup \{j\}}$. Then, due to F2, we have $\mathcal{M}^{A'}_i = \mathcal{M}^{A'}_j$ for all $A' \subseteq C \setminus \{i, j, a\}$, especially $\mathcal{M}_j^{C \setminus \{a\}} = \mathcal{M}_i^{C \setminus \{a\}}$. Furthermore, for $B = C \setminus \{i, j\}$ we have $v_{B \cup \{i\}} > v_{B \cup \{j\}}$. Thus, $v_{C \setminus \{j\}} > v_{C \setminus \{i\}}$, so for a as defined above, we have 
\begin{align}
    \mathcal{M}_a^{C \setminus \{j\}} = v_{C \setminus \{j\}} > v_{C \setminus \{i\}} \geq \mathcal{M}_a^{C \setminus \{i\}}.
\end{align}
Again, together with F6, we have:
\begin{align*}
    &\mathcal{M}_i^C=&\mathcal{M}_a^C - \mathcal{M}_a^{C \setminus \{i\}} + \mathcal{M}_i^{C \setminus \{a\}}& \\
    >&&\mathcal{M}_a^C - \mathcal{M}_a^{C \setminus \{j\}} + \mathcal{M}_j^{C \setminus \{a\}} & = \mathcal{M}_j^C.
\end{align*}\\

\textbf{Case 2: } $\mathcal{M}_i^{C \setminus \{j\}} = v_{C \setminus \{j\}}$.

First, suppose $\mathcal{M}_i^{C \setminus \{j\}} > \mathcal{M}_j^{C \setminus \{i\}}$. This directly gives 
\begin{align}
    \mathcal{M}^C_i > \mathcal{M}^C_i - \mathcal{M}_i^{C \setminus \{j\}} + \mathcal{M}_j^{C \setminus \{i\}}
\end{align}
and thus F3 holds for $C$.

Now, suppose $\mathcal{M}_i^{C \setminus \{j\}} = \mathcal{M}_j^{C \setminus \{i\}}$. Because of $\mathcal{M}_i^{C \setminus \{j\}} = v_{C \setminus \{j\}}$ and $v_{C \setminus \{i\}} \geq \mathcal{M}_j^{C \setminus \{i\}}$ (due to R2), this would imply $v_{C \setminus \{j\}} = v_{C \setminus \{i\}}$

Then $B \neq C \setminus \{i, j\}$ and must be some $a \in C \setminus \{i, j\}$, such that for $B \subseteq C \setminus \{i, j, a\}$, $B \neq \emptyset$ and $v_{B \cup \{i\}} > v_{B \cup \{j\}}$. Then, because due to the inductive hypothesis the algorithm fulfills F3 for $C \setminus \{a\}$, $\mathcal{M}_i^{C \setminus \{a\}} > \mathcal{M}_j^{C \setminus \{a\}}$ must hold. Since due to F5
\begin{align}
     \mathcal{M}_i^{C \setminus \{a\}} = \mathcal{M}_j^{C \setminus \{a\}} - \mathcal{M}_j^{C \setminus \{a, i\}} + \mathcal{M}_i^{C \setminus \{a, j\}},
\end{align}
$\mathcal{M}_j^{C \setminus \{a, i\}} < \mathcal{M}_i^{C \setminus \{a, j\}}$ must hold as well. R5 yields the following two equations: 
\begin{equation}
    \mathcal{M}^{C \setminus \{j\}}_a = \mathcal{M}^{C \setminus \{j\}}_i - \mathcal{M}^{C \setminus \{a, j\}}_i + \mathcal{M}^{C \setminus \{i, j\}}_a
\end{equation}
and
\begin{equation}
    \mathcal{M}^{C \setminus \{i\}}_a = \mathcal{M}^{C \setminus \{i\}}_j - \mathcal{M}^{C \setminus \{a, i\}}_j + \mathcal{M}^{C \setminus \{i, j\}}_a
\end{equation}
Using $\mathcal{M}_i^{C \setminus \{j\}} = \mathcal{M}_j^{C \setminus \{i\}}$ and $\mathcal{M}_j^{C \setminus \{a, i\}} < \mathcal{M}_i^{C \setminus \{a, j\}}$ this yields
\begin{equation}
    \mathcal{M}^{C \setminus \{j\}}_a > \mathcal{M}^{C \setminus \{i\}}_a.
\end{equation}
Due to the inductive hypothesis, $\mathcal{M}_i^{C \setminus \{a\}} > \mathcal{M}_j^{C \setminus \{a\}}$. Again, together with F6, we have:
\begin{align*}
    &\mathcal{M}_i^C=&\mathcal{M}_a^C - \mathcal{M}_a^{C \setminus \{i\}} + \mathcal{M}_i^{C \setminus \{a\}}& \\
    >&&\mathcal{M}_a^C - \mathcal{M}_a^{C \setminus \{j\}} + \mathcal{M}_j^{C \setminus \{a\}} & = \mathcal{M}_j^C.
\end{align*}\\

Finally, note that $\mathcal{M}_i^{C \setminus \{j\}} < \mathcal{M}_j^{C \setminus \{i\}}$ is impossible, since $\mathcal{M}_i^{C \setminus \{j\}} = v_{C \setminus \{j\}} \geq v_{C \setminus \{i\}} \geq \mathcal{M}_j^{C \setminus \{i\}}$ due to R2.\\

Therefore, the algorithm fulfills F3 for $C$. By the principle of mathematical induction, algorithm \ref{solfun} fulfills F2 for all $C \subseteq N$ with $\lvert C \rvert = s$ and $s \leq |N|$.\\


\subsection{Proof of F4}
\begin{tabularx}{\linewidth}{l X}
    F4 & \textbf{Strict Monotonicity:} Let $v$ and $v'$ denote any two value functions over all coalitions $C \subseteq N$ and $\mathcal{M}(v)^C_i$ and $\mathcal{M}(v')^C_i$ be the corresponding values of rewards received by player $i$ in coalition $C$. For all $C \in N$ and all $i \in C$,
    {\begin{align}
        \begin{split}
            \left(v'_C > v_C \right) & \land\\
            \left( \forall A \subseteq C \setminus \{i\}: v'_{A \cup \{i\}} \geq v_{A \cup \{i\}} \right) & \land\\
            \left( \forall D \subseteq C \setminus \{i\}: v'_D = v_D \right)\\
            \Rightarrow  \mathcal{M}(v')^C_i > \mathcal{M}(v)^C_i.
        \end{split}
    \end{align}}
\end{tabularx}\\

Let $v$ and $v'$ denote any two arbitrary but fixed value functions over all coalitions $C \subseteq N$ and $\mathcal{M}(v)^C_i$ that are monotonic with $v_{\emptyset} = v'_{\emptyset}$. We show that the algorithm \ref{solfun} fulfills F3 by induction over the size of $C$.\\

\textbf{Base case:} Let $C \subseteq N$ with $\lvert C \rvert \leq 1$. It is easy to see that F4 holds here. In the case of one-member coalitions F4 simplifies to: $\forall C \subseteq N$ with $\lvert C \rvert = 1$ and $i \in C$,
{\begin{align}
        \begin{split}
            \left( v'_i > v_i \right) & \land\\
            \left( v'_{\emptyset} = v_{\emptyset} \right) \\
            \Rightarrow  \mathcal{M}(v')^C_i > \mathcal{M}(v)^C_i.
        \end{split}
    \end{align}}
since here $C \setminus \{i\} = \emptyset$. The axiom holds here, since under our algorithm due to line 4, $\mathcal{M}(v')^C_i = v'_i > v_i = \mathcal{M}(v)^C_i$. For $C = \emptyset$, this axiom does not apply.\\

\textbf{Inductive hypothesis:} Suppose, algorithm \ref{solfun} fulfills F4 for all coalitions $C \subseteq N$ with $\lvert C \rvert < s$ up to some size $s$, $s \geq 1$.\\

\textbf{Inductive step:} Let $C \subseteq N$ with $\lvert C \rvert=s$. Suppose for some $i \in C$, $v'_C > v_C$ and $\forall A \subseteq C \setminus \{i\}: v'_{A \cup \{i\}} \geq v_{A \cup \{i\}}$ and $\forall D \subseteq C \setminus \{i\}: v'_D = v_D$. We will show that $\mathcal{M}(v')^C_i > \mathcal{M}(v)^C_i$ holds.

As we have already seen, our algorithm fulfills R3, R5 and F5. Then due to theorem \ref{unique}, we have $\mathcal{M}(v')^{D}_i = \mathcal{M}(v)^{D}_i$ for all $D \subset C \setminus \{i\}$. $\mathcal{M}(v')^{C}_i$ is assigned either in line 18 or 20. In line 18 we would have 
\begin{align}
    \mathcal{M}(v')^{C}_i = v'_C > v_C \geq \mathcal{M}(v)^{C}_i
\end{align}
(the last inequality holds due to R3). In line 18, $k$ would be the player for which $\mathcal{M}(v')^C_k = v'_C$ and then
\begin{align}
    \begin{split}
    \mathcal{M}(v')^{C}_i & = \mathcal{M}(v')^C_k - \mathcal{M}(v')^{C \setminus \{i\}}_k + \mathcal{M}(v')^{C \setminus \{k\}}_i\\
    & = v_C - \mathcal{M}(v')^{C \setminus \{i\}}_k + \mathcal{M}(v')^{C \setminus \{k\}}_i\\
    & > v_C - \mathcal{M}(v)^{C \setminus \{i\}}_k + \mathcal{M}(v)^{C \setminus \{k\}}_i\\
    &= \mathcal{M}(v)^{C}_i
    \end{split}
\end{align}
because, as mentioned above, $\mathcal{M}(v')^{C \setminus \{i\}}_k = \mathcal{M}(v)^{C \setminus \{i\}}_k$. Furthermore, due to theorem \ref{unique} and the inductive hypothesis, $\mathcal{M}(v')^{C \setminus \{k\}}_i \geq \mathcal{M}(v)^{C \setminus \{k\}}_i$ has to hold, because either for all $A \subseteq C \setminus \{k\}$: $v'_A = v_A$ and thus $\mathcal{M}(v')^A_i > \mathcal{M}(v)^A_i$. Or, for some $A \subseteq C \setminus \{k\}$: $v'_A > v_A$. Then, due to the inductive hypothesis, $\mathcal{M}(v')^A_i > \mathcal{M}(v)^A_i$ and for all $B$ with $A \subseteq B \subseteq C$: $\mathcal{M}(v')^B_i \geq \mathcal{M}(v)^B_i$.

In conclusion, algorithm \ref{solfun} fulfills F4 for all $C \subseteq N$ $\lvert C \rvert=s$. By the principle of mathematical induction, F4 holds for all $s \leq |N|$.


\begin{thebibliography}{24}
\providecommand{\natexlab}[1]{#1}
\providecommand{\url}[1]{\texttt{#1}}
\expandafter\ifx\csname urlstyle\endcsname\relax
  \providecommand{\doi}[1]{doi: #1}\else
  \providecommand{\doi}{doi: \begingroup \urlstyle{rm}\Url}\fi

\bibitem[Banzhaf(1965)]{banzhaf65weightedVoting}
J.~Banzhaf.
\newblock Weighted voting doesn't work: {A} mathematical analysis.
\newblock \emph{Rutgers Law Review}, 19\penalty0 (2):\penalty0 317--343, 1965.

\bibitem[Brandt et~al.(2016)Brandt, Conitzer, Endriss, Lang, and
  Procaccia]{brandt16handbook}
F.~Brandt, V.~Conitzer, U.~Endriss, J.~Lang, and A.~Procaccia, editors.
\newblock \emph{Handbook of Computational Social Choice}.
\newblock Cambridge University Press, 2016.

\bibitem[Chalkiadakis et~al.(2011)Chalkiadakis, Elkind, and
  Wooldridge]{chalkiadakis2011computational}
G.~Chalkiadakis, E.~Elkind, and M.~Wooldridge.
\newblock \emph{Computational aspects of cooperative game theory}.
\newblock Morgan \& Claypool Publishers, 2011.

\bibitem[Conitzer and Oesterheld(2023)]{conitzer23foundationsOFCooperativeAI}
V.~Conitzer and C.~Oesterheld.
\newblock Foundations of cooperative ai.
\newblock \emph{Proceedings of the AAAI Conference on Artificial Intelligence},
  37\penalty0 (13):\penalty0 15359--15367, Jul. 2023.
\newblock \doi{10.1609/aaai.v37i13.26791}.
\newblock URL \url{https://ojs.aaai.org/index.php/AAAI/article/view/26791}.

\bibitem[{Conitzer} et~al.(2024){Conitzer}, {Freedman}, {Heitzig}, {Holliday},
  {Jacobs}, {Lambert}, {Moss{\'e}}, {Pacuit}, {Russell}, {Schoelkopf},
  {Tewolde}, and {Zwicker}]{conitzer24socialChoice}
V.~{Conitzer}, R.~{Freedman}, J.~{Heitzig}, W.~H. {Holliday}, B.~M. {Jacobs},
  N.~{Lambert}, M.~{Moss{\'e}}, E.~{Pacuit}, S.~{Russell}, H.~{Schoelkopf},
  E.~{Tewolde}, and W.~S. {Zwicker}.
\newblock {Social Choice Should Guide AI Alignment in Dealing with Diverse
  Human Feedback}.
\newblock \emph{arXiv e-prints}, art. arXiv:2404.10271, 2024.
\newblock \doi{10.48550/arXiv.2404.10271}.

\bibitem[Filter et~al.(2024)Filter, M{\"o}ller, and
  {\"O}z{\c{c}}ep]{filter2024mechanisms}
B.~Filter, R.~M{\"o}ller, and {\"O}.~L. {\"O}z{\c{c}}ep.
\newblock Mechanisms for data sharing in collaborative causal inference.
\newblock In \emph{German Conference on Artificial Intelligence (K{\"u}nstliche
  Intelligenz)}, pages 86--98. Springer, 2024.

\bibitem[Gonzalez and Grabisch(2015)]{gonzalez15preserving}
S.~Gonzalez and M.~Grabisch.
\newblock Preserving coalitional rationality for non-balanced games.
\newblock \emph{International Journal of Game Theory}, 44\penalty0
  (3):\penalty0 733--760, 2015.
\newblock \doi{10.1007/s00182-014-0451-9}.
\newblock URL \url{https://doi.org/10.1007/s00182-014-0451-9}.

\bibitem[Gonzalez and Grabisch(2016)]{gonzalez16multicoalitional}
S.~Gonzalez and M.~Grabisch.
\newblock Multicoalitional solutions.
\newblock \emph{Journal of Mathematical Economics}, 64:\penalty0 1--10, 2016.
\newblock ISSN 0304-4068.
\newblock \doi{https://doi.org/10.1016/j.jmateco.2016.02.006}.
\newblock URL
  \url{https://www.sciencedirect.com/science/article/pii/S0304406816000148}.

\bibitem[Grabisch(2016)]{grabisch16bases}
M.~Grabisch.
\newblock \emph{Bases and Transforms of Set Functions}, pages 215--231.
\newblock Springer International Publishing, Cham, 2016.
\newblock ISBN 978-3-319-28808-6.
\newblock \doi{10.1007/978-3-319-28808-6_13}.
\newblock URL \url{https://doi.org/10.1007/978-3-319-28808-6_13}.

\bibitem[Hendrycks(2025)]{hendrycks25AiSafety}
D.~Hendrycks.
\newblock \emph{Introduction to AI Safety, Ethis, and Society}.
\newblock Routledge, 2025.

\bibitem[Karimireddy et~al.(2022)Karimireddy, Guo, and
  Jordan]{karimireddy2022mechanisms}
S.~P. Karimireddy, W.~Guo, and M.~I. Jordan.
\newblock {Mechanisms that Incentivize Data Sharing in Federated Learning}.
\newblock Papers 2207.04557, arXiv.org, 2022.

\bibitem[Lehrer(1988)]{lehrer1988axiomatization}
E.~Lehrer.
\newblock An axiomatization of the banzhaf value.
\newblock \emph{International Journal of Game Theory}, 17:\penalty0 89--99,
  1988.

\bibitem[Maschler and Peleg(1966)]{maschler1966characterization}
M.~Maschler and B.~Peleg.
\newblock A characterization, existence proof and dimension bounds for the
  kernel of a game.
\newblock \emph{pacific Journal of Mathematics}, 18\penalty0 (2):\penalty0
  289--328, 1966.

\bibitem[Qiao et~al.(2023)Qiao, Xu, and Low]{qiao23collaborative}
R.~Qiao, X.~Xu, and B.~K.~H. Low.
\newblock Collaborative causal inference with fair incentives.
\newblock In A.~Krause, E.~Brunskill, K.~Cho, B.~Engelhardt, S.~Sabato, and
  J.~Scarlett, editors, \emph{Proceedings of the 40th International Conference
  on Machine Learning}, volume 202 of \emph{Proceedings of Machine Learning
  Research}, pages 28300--28320. PMLR, 23--29 Jul 2023.
\newblock URL \url{https://proceedings.mlr.press/v202/qiao23a.html}.

\bibitem[Russell(2022)]{russell22artificial}
S.~Russell.
\newblock \emph{Artificial Intelligence and the Problem of Control}, pages
  19--24.
\newblock Springer International Publishing, Cham, 2022.
\newblock ISBN 978-3-030-86144-5.
\newblock \doi{10.1007/978-3-030-86144-5_3}.
\newblock URL \url{https://doi.org/10.1007/978-3-030-86144-5_3}.

\bibitem[Shapley(1953{\natexlab{a}})]{shapley1953value}
L.~S. Shapley.
\newblock A value for n-person games.
\newblock In H.~W. Kuhn and A.~W. Tucker, editors, \emph{Contributions to the
  Theory of Games II}, pages 307--317. Princeton University Press, Princeton,
  1953{\natexlab{a}}.

\bibitem[Shapley(1953{\natexlab{b}})]{shapley53value}
L.~S. Shapley.
\newblock \emph{17. A Value for n-Person Games}, pages 307--318.
\newblock Princeton University Press, Princeton, 1953{\natexlab{b}}.
\newblock ISBN 9781400881970.
\newblock \doi{doi:10.1515/9781400881970-018}.
\newblock URL \url{https://doi.org/10.1515/9781400881970-018}.

\bibitem[Shoham and Leyton-Brown(2008)]{shoham2008multiagent}
Y.~Shoham and K.~Leyton-Brown.
\newblock \emph{{Multiagent systems: Algorithmic, game-theoretic, and logical
  foundations}}.
\newblock Cambridge University Press, 2008.

\bibitem[Sim et~al.(2020)Sim, Zhang, Chan, and Low]{sim2020collaborative}
R.~H.~L. Sim, Y.~Zhang, M.~C. Chan, and B.~K.~H. Low.
\newblock Collaborative machine learning with incentive-aware model rewards.
\newblock In \emph{International conference on machine learning}, pages
  8927--8936. PMLR, 2020.

\bibitem[Wang and Jia(2023)]{wang2023data}
J.~T. Wang and R.~Jia.
\newblock Data banzhaf: A robust data valuation framework for machine learning.
\newblock In \emph{International Conference on Artificial Intelligence and
  Statistics}, pages 6388--6421. PMLR, 2023.

\bibitem[Wang et~al.(2020)Wang, Rausch, Zhang, Jia, and
  Song]{wang2020principled}
T.~Wang, J.~Rausch, C.~Zhang, R.~Jia, and D.~Song.
\newblock A principled approach to data valuation for federated learning.
\newblock \emph{Federated Learning: Privacy and Incentive}, pages 153--167,
  2020.

\bibitem[Xu et~al.(2021)Xu, Lyu, Ma, Miao, Foo, and Low]{xu2021gradient}
X.~Xu, L.~Lyu, X.~Ma, C.~Miao, C.~S. Foo, and B.~K.~H. Low.
\newblock Gradient driven rewards to guarantee fairness in collaborative
  machine learning.
\newblock \emph{Advances in Neural Information Processing Systems},
  34:\penalty0 16104--16117, 2021.

\bibitem[Young(1985)]{young1985monotonic}
H.~P. Young.
\newblock Monotonic solutions of cooperative games.
\newblock \emph{International Journal of Game Theory}, 14\penalty0
  (2):\penalty0 65--72, 1985.

\bibitem[Zou et~al.(2020)Zou, van~den Brink, and Funaki]{zou2020sharing}
Z.~Zou, R.~van~den Brink, and Y.~Funaki.
\newblock Sharing the surplus and proportional values.
\newblock \emph{Theory and Decision}, pages 1--33, 2020.

\end{thebibliography}
\end{document}